\documentclass[10pt]{article}
\usepackage{amsmath,amssymb,amstext,mathtools,array,url,bm,graphicx,color,epsfig}
\usepackage{fullpage,setspace}
\usepackage{authblk}
\usepackage{filecontents}
\usepackage{natbib}
\usepackage{lineno}
\usepackage[hidelinks]{hyperref}
\usepackage{subcaption}
\usepackage{float}
\usepackage[flushleft]{threeparttable}

\newtheorem{theorem}{Theorem}
\newtheorem{definition}[theorem]{Definition}
\newtheorem{proposition}[theorem]{Proposition}
\newenvironment{proof}[1][Proof]{\noindent\textbf{#1.} }{\ \rule{0.5em}{0.5em}}

\def\BE{\begin{equation}}
\def\EE{\end{equation}}

\DeclareSymbolFont{largesymbolsA}{U}{txexa}{m}{n}
\DeclareMathSymbol{\varprod}{\mathop}{largesymbolsA}{16}

\begin{document}

\title{\bf Comparative analysis of the structures and outcomes of geophysical flow models and modeling assumptions using uncertainty quantification}

\author[1,2]{\small Abani Patra}
\author[3,2]{Andrea Bevilacqua}
\author[1]{Ali Akhavan-Safaei}
\author[4]{E. Bruce Pitman}
\author[3]{Marcus Bursik}
\author[3]{David Hyman}

\affil[1]{\small\textit{Department of Mechanical and  Aerospace Engineering, University at Buffalo, Buffalo, NY} }
\affil[2]{\textit{Computational Data Science and Engineering, University at Buffalo, Buffalo, NY}}
\affil[3]{\textit{Department of Earth Sciences, University at Buffalo, Buffalo, NY} }
\affil[4]{\textit{Department of Materials Design and Innovation, University at Buffalo, NY}}

\date{\texttt{\{abani,abevilac,aliakhav,pitman,mib,davidhym\}@buffalo.edu}}

\maketitle

\abstract

We present a new statistically driven method for analyzing the modeling of geophysical flows. Many models have been advocated by different modelers for such flows incorporating different modeling assumptions. Limited and sparse observational data on the modeled phenomena usually does not permit a clean discrimination among models for fitness of purpose, and, heuristic choices are usually made, especially for critical predictions of behavior that has not been experienced. We advocate here a methodology for characterizing models and the modeling assumptions they represent, using a statistical approach over the full range of applicability of the models. Such a characterization may then be used to decide the appropriateness of a model and modeling assumption for use.
We present our method by comparing three different models arising from different rheology assumptions, and the data show unambiguously the performance of the models across a wide range of possible flow regimes. This comparison is facilitated by the recent development of the new release of our TITAN2D mass flow code that allows choice of multiple rheologies The quantitative and probabilistic analysis of contributions from different modeling  assumptions in the models is particularly illustrative of the impact of the assumptions. Knowledge of which assumptions dominate, and, by how much, is illustrated in two different case studies: a small scale inclined plane with a flat runway, and the large scale topography on the SW slope of Volc\'{a}n de Colima (MX). A simple model performance evaluation completes the presentation.

\section{Models and assumptions}

Complex systems with sparse observations like large scale geophysical mass flows are often represented in the literature by many models, e.g., \cite{Kelfoun2011}. It is often difficult if not impossible to decide which of these models are appropriate for a particular analysis. Ready availability of many models as reusable software tools makes it the users burden to select one appropriate for their purpose.
For example, the $\mathrm{4^{\mathrm{th}}}$ release of TITAN2D\footnote{available from vhub.org} offers multiple rheology options in the same code base. The availability of three distinct models for similar phenomena in the same tool provides us the ability to directly compare  outputs and internal variables in all the three models and control for usually difficult to quantify effects like numerical solution procedures, input ranges and computer hardware. Given a particular problem for which predictive analysis is planned, the information generated  could possibly be used to guide model choice and/or processes for integrating information from multiple models \citep{Bongard2007}. However, as we have discovered such comparison requires a more careful understanding of each model and its constituents and a well organized process for such comparison. We begin with precise, albeit limited definitions
of models and their constituents.

\cite{Gilbert91} defines science as a process of constructing predictive conceptual models where these models represent
consistent predictive relations in target systems. A simpler definition of a model that is more appropriate for the context of this study is that:

\begin{quote}{\it A model is a representation of a postulated relationship among inputs and outputs of a system, informed by observation, and based on an hypothesis that best explains the relationship.}\end{quote} The definition captures two of the most important characteristics
\begin{itemize}
\item models depend on a {\it hypothesis}, and,
\item models use the {\it data from observation} to validate and refine the hypothesis.
\end{itemize}
Errors and uncertainty in the data and limitations in the hypothesis (usually a tractable and computable mathematical construct articulating beliefs like proportionality, linearity, etc.) are immediate challenges that must be overcome, to construct useful and credible models.

A model is most useful in predicting the behavior of a system for unobserved inputs, and, interpreting or explaining of the system's behavior. Since models require an hypothesis, it follows  that a model is a formulation of a belief about the data. The immediate consequence of this is that the model may not be reliable about such prediction, since the {\it subjectivity of the belief} can never be completely eliminated \citep{Kennedy2001, Higdon2004}, even when sufficient care is taken to use all the available data and information. Second, the data at hand may not provide enough information about the system to characterize its behavior at the desired prediction. This {\it data inadequacy} is rarely characterized even for verified and validated models. 

The consequence of this lack of knowledge and limited data is the multiplicity of beliefs about the complex system being modeled and a profusion of models based on different modeling assumptions and data use. These competing models lead to much debate among scientists. Principles like ``Occam's razor" and Bayesian statistics \citep{Farrell2015} provide some guidance, but simple robust approaches that allow the testing of models for fitness need to be developed. We present in this paper a simple data driven approach to discriminate among models and the modeling assumptions implicit in each model, given a range of phenomena to be studied. We illustrate the approach by work on geophysical mass flows.

An assumption is a simple concept -- any atomic postulate about relationships among quantities under study. Models are compositions of many such assumptions. The study of models is, thus, implicitly a study of these assumptions and their composability and applicability in a particular context. Sometimes a good model contains a useless assumption that may be removed, sometimes a good assumption could be added to a different model - these are usually subjective choices, not data driven. Moreover, the correct assumptions may change as the system evolves, making {\em model choice} more difficult.

The statistically driven method introduced in this study for analyzing complex models provides extensive and quantitative information. Geophysical flow modeling usually compares simulation to observation, and fits the model parameters using the solution of a regularized inverse problem. Nevertheless, this is not always sufficient to solve forecasting problems, in which the range of possible flows might not be limited to a single type and scale of flow. Our approach is different, and evaluates the statistics of a range of flows, produced by the couple $\left(M, P_M\right)$ - i.e. a model $M$ and a probability distribution for its parameters $P_M$.

New quantitative information can solve classical qualitative problems, either model-model or model-observation comparison. The mean plot represents the average behavior of flows in the considered range, and provides the same type of data that is provided by a single simulation. Moreover, the uncertainty ranges generate additional pieces of information that often highlight the differences between the models.

The rest of the paper will define our approach and illustrate it using three models for large scale mass flows incorporated in our large scale mass flow simulation framework  TITAN2D \citep{Patra2005,Patra2006, Yu2009, Aghakhani2016}.
So far, TITAN2D has been successfully applied to the simulation of different geophysical mass flows with specific characteristics \citep{Sheridan2005, Rupp2006, Norini2008, Charbonnier2009, Procter2010, Sheridan2010, Sulpizio2010, Capra2011}. Several studies involving TITAN2D were also directed towards a statistical study of geophysical flows, focusing on uncertainty quantification \citep{Dalbey2008, Dalbey2009, Stefanescu2012b, Stefanescu2012a}, or on the more efficient production of hazard maps \citep{Bayarri2009, Spiller2014, Bayarri2015, Ogburn2016}.

In particular, we initially provide a traditional type of analysis, summarizing the general features that differ among the model outputs. However, this is performed in a probabilistic framework, oriented to extrapolation and forecast. After that, and more significantly, we describe and compare the general features of the newly introduced \emph{contributing variables} in the models, through the new concepts of dominance factors and expected contributions. This is a type of analysis enabled by our approach, that allows us to evaluate modeling assumptions and their relative importance.

\section{Analysis of modeling assumptions and models }
\subsection{Analysis Process}
Let us define $\left(M(A), P_{M(A)}\right)$, where $A$ is a set of assumptions, $M(A)$ is the model which combines those assumptions, and $P_M$ is a probability distribution in the parameter space of $M$. While the support of $P_M$ can be restricted to a single value by solving an inverse problem for the optimal reconstruction of a particular flow, this is not desirable if we are interested in the general predictive capabilities of the model, where we are interested in the outcomes over a whole range.

Our problem cannot be solved using classical sensitivity analysis (e.g. \cite{Saltelli2010}, \cite{Weirs2012}), which decomposes the variance of model output with respect to the input parameters. Indeed, model assumptions cannot be seen as input parameters, because they are related to the terms in the governing equations. These terms can be seen as random variables depending on the inputs, but they have an unknown probability distribution and are not independent. In the sequel we will define the new concepts of Dominance Factors and Expected Contributions to cope with this problem. We summarize our analysis process in thre steps.
\paragraph{Stage 1: Parameter Ranges} In this study, we assume:
$$P_M\left(p_1,\dots,p_{N_M}\right)\sim \bigotimes_{i=1}^{N_M} Unif(a_{i,M},b_{i,M}),$$
where $N_M$ is the number of parameters of $M$. This is not restrictive, and in case of correlation between the commonly used parameters $(\hat p_j)_{j=1,\hat N_M}$, or non-uniform distributions, we can always define a function $g$ such that $g[(p_i)]=(\hat p_j)$, and the $(p_i)$ are independent and uniformly distributed. In particular, we choose these parameter ranges using information gathered from the literature about the physical meaning of those values, together with a preliminary testing for physical consistency of model outcomes and range of inputs/outcomes of interest. {This step is critical, because if the statistical comparison is dominated by trivial macroscopic differences, it cannot focus on the rheology details.}  In the preparation of hazard analysis, expert elicitation processes can be
used to ensure that the studies correctly account for all anticipated and possible flow regimes.

\paragraph{Stage 2 Simulations and Data Gathering}
For each model $M$, we produce datasets of \emph{model inputs}, \emph{contributing variables} and \emph{model outputs}. The \emph{contributing variables} include quantities in the model evaluation that are ascribable to specific assumptions $A_i$. These are usually not observed as outputs from the model. For example in momentum balances of complex flow calculations these could be values of different source terms, dissipation terms and inertia terms. The \emph{model outputs} include explicit outcomes e.g., for flow calculations these could be flow height, lateral extent, area, velocity, acceleration, and derived quantities such as Froude number $Fr$. In general, we use a Monte Carlo simulation, sampling the model inputs and obtaining a family of graphs plotting the expectation of the contributing variables and model outputs. We also include their 5$^{\mathrm{th}}$ and 95$^{\mathrm{th}}$ percentiles. Our sampling technique of the input variables is based on the Latin Hypercube Sampling idea, and in particular, on the improved space-filling properties of the orthogonal array-based Latin Hypercubes (see Appendix \ref{A-1}). The volume of data generated is likely to be large but modern computing and data handling equipment readily available to most modeling researchers \footnote{We thank the University at Buffalo Center for Computing Research} in university and national research facilities are more than adequate.

\paragraph{Stage 3: Results Analysis} These and other statistics can now be compared to determine the need for different modeling assumptions and the relative merits of different models. Thus, analysis of the data gathered over the entire range of flows for the state variables and outcomes leads to a quantitative basis for accepting or rejecting particular assumptions or models for specific outcomes.

\subsection{Monte Carlo Process and Statistical Analysis}
In our study, the flow range is defined by establishing boundaries for inputs like flow volume or rheology coefficients characterizing the models.  Latin Hypercube Sampling is performed over $[0,1]^d$ where $d$ depends on the number of  input parameters (see Appendix \ref{A-1}). Those dimensionless samples are linearly  mapped to fill the required intervals. Section \ref{lhs_des_colima} provides examples of Latin Hypercube design in the three models that are targets of this study, with respect to their commonly used parameters.

Following the simulations, we generate data for each sample run and each outcome and contributing variables $f(\underline{\textbf x},t)$ calculated as a function of time on the elements of the computational grid. This analysis generates tremendous volume of data which must then be analyzed using statistical methods for summative impact. The contributing variables in this case are the mass and force terms in the conservation laws defined above.

We devise many statistical measures for analyzing the data. For instance, let $(F_i(\underline{\textbf x},t))_{i=1,\dots, 4}$ be an array of force terms, where $\underline{\textbf x}\in \mathbf R^d$ is a spatial location, and $t\in T$ is a time instant. The degree of contribution of those force terms to the flow dynamics can be significantly variable in space and time, and we define the \emph{dominance factors} $[p_j(\underline{\textbf x},t)]_{j=1,\dots, k}$, i.e., the probability of each $F_j(\underline{\textbf x},t)$ to be the dominant force. Those probabilities provide insight into the dominance of a particular source or dissipation term on the model dynamics. Each term is identified with a particular modeling assumption. We remark that we focus on the modulus of the forces and hence we cope with scalar terms. It is also important to remark that all the forces depend on the input variables, and they can be thus considered as random variables. Furthermore, these definitions are general and could be applied to any set of contributing variables, and not only to the force terms.

\begin{definition}[Dominance factors]
Let $(F_i)_{i\in I}$ be random variables on $(\Omega, \mathcal F, P_M)$. Then, $\forall i$, the dominant variable is defined as:
$$\Phi:=\max_i |F_i|.$$
In particular, for each $j \in I$, the dominance factors are defined as:
$$p_j:=P_M\left\{\Phi=|F_j|\right\}.$$
\end{definition}

We remark that the dominant variable $\Phi$ is also a random variable, and in particular it is a stochastic process parameterized in time and space. Moreover, we define the \emph{random contributions}, an additional tool that we use to compare the different force terms, following a less restrictive approach than the dominance factors. They are obtained dividing the force terms by the dominant force $\Phi$, and hence belong to $[0,1]$.

\begin{definition}[Expected contributions]
Let $(F_i)_{i\in I}$ be random variables on $(\Omega, \mathcal F, P_M)$. Then, $\forall i$, the random contribution is defined as:
$$C_i:=\left\{
\begin{array}{ll}
      \frac{F_i}{\Phi}, & \textrm{if }\Phi\neq 0; \\
      0, & \hbox{otherwise.}
    \end{array}
  \right.$$
where $\Phi$ is the dominant variable. Thus, $\forall i$, the expected contributions are defined by $E\left[C_i\right]$.
\end{definition}

In particular, for a particular location $x$, time $t$, and parameter sample $\omega$, we have $C_i(\underline{\textbf x},t,\omega)=0$ if there is no flow or all the forces are null. The expectation of $C_i$ is reduced by the chance of $F_i$ being small compared to the other terms, or by the chance of having no flow in $(\underline{\textbf x},t)$. The meaning of the random variable $\Phi(\underline{\textbf x},t)$ is explained by the dominance factors, and hence they can be used to define a further statistical decomposition of the random contributions, as detailed in the Appendix \ref{A-2}.

\section{Modeling of geophysical mass flows}\label{subsec:FlowTypes}
Dense large scale granular avalanches are a complex class of flows with physics that has often been poorly captured by models that are computationally tractable. Sparsity of actual flow data (usually only a posteriori deposit information is available), and large uncertainty in the mechanisms of initiation and flow propagation, make the modeling task challenging, and a subject of much continuing interest. Models that appear to represent the physics well in certain flows, may turn out to be poorly behaved in others, due to intrinsic physical, mathematical or numerical issues. Nevertheless, given the large implications on life and property, many models with different modeling assumptions have been proposed. For example in \citep{Iverson1997, Iverson2001, Denlinger2001, Pitman2003a, Denlinger2004, Iverson2004}, the depth-averaged model was applied in the simulation of test geophysical flows in large scale experiments. Several studies were specifically devoted to the modeling of volcanic mass flows \citep{Bursik2005,Kelfoun2005,Macias2008,Kelfoun2009,Charbonnier2013}. In fact, volcanos are great sources for a rich variety of geophysical flow types and provide field data from past flow events.

Modeling in this case proceeds by first assuming that the laws of mass and momentum conservation hold for properly defined system boundaries. The scale of these flows -- very long and wide with small depth led to the first most generally accepted assumption -- shallowness \citep{SavageHutter1989}. This allows an integration through the depth to obtain simpler and more computationally tractable equations. This is the next of many assumptions that have to be made. Both of these are fundamental assumptions which can be tested in the procedure we established above. Since, there is a general consensus and much evidence in the literature of the validity of these assumptions we defer analysis of these to future work.

The depth-averaged Saint-Venant equations that result are:
\begin{eqnarray}
\label{eq:D_A}
\frac{\partial h}{\partial t} +
\frac{\partial}{\partial x}(h \bar{u}) +
\frac{\partial}{\partial y}(h\bar{v}) &=& 0 \nonumber \\
\frac{\partial}{\partial t} (h\bar{u}) +
\frac{\partial}{\partial x}\left(h\bar{u}^2 + \frac{1}{2}k g_{z}h^2\right) + \frac{\partial}{\partial y}(h\bar{u}\bar{v}) &=& S_{x}\\
\frac{\partial}{\partial t} (h\bar{v}) +
\frac{\partial}{\partial x}(h\bar{u}\bar{v}) +
\frac{\partial}{\partial y}\left(h\bar{v}^2 + \frac{1}{2}k g_{z}h^2\right) &=& S_{y} \nonumber
\end{eqnarray}
Here the Cartesian coordinate system is aligned such that $z$ is normal to the surface; $h$ is the flow height in the $z$ direction; $h\bar{u}$ and $h\bar{v}$ are respectively the components of momentum in the $x$ and $y$ directions; and $k$ is the coefficient which relates the lateral stress components, $\bar{\sigma}_{xx}$ and $\bar{\sigma}_{yy}$, to the normal stress component, $\bar{\sigma}_{zz}$. The definition of this coefficient depends on the constitutive model of the flowing material we choose. Note that $\frac{1}{2} k g_z h^2$ is the contribution of depth-averaged pressure to the momentum fluxes. $S_x$ and $S_y$ are the sum local stresses: they include the gravitational driving forces, the basal friction force resisting to the motion of the material, and additional forces specific of rheology assumptions.

The final class of assumptions are the assumptions on the rheology of the flows -- in particular in this context assumptions used to model different dissipation mechanisms embedded in $S_x, S_y$ that lead to a plethora of models with much controversy on the most suitable model. We focus here on three models derived from different assumptions for essentially the same class of flows.

\subsection{Overview of the models}\label{subsec:Models}
In the three following sections, we briefly describe \emph{Mohr-Coulomb} (MC), \emph{Pouliquen-Forterre} (PF) and \emph{Voellmy-Salm} (VS) models. Models based on additional heterogeneous assumptions are possible, either more complex \citep{PitmanLe2005,Iverson2014} or more simple \citep{DadeHuppert1998}. We decided to focus on these three because of their historical relevance. Moreover, if the degree of complexity in the models is significantly different, model comparison should take into account that, but this is outside the purpose of this study \citep{Farrell2015}.

\subsubsection{Mohr-Coulomb}\label{MCM}
Based on the long history of studies in soil mechanics \citep{Rankine1857,DruckerPage52}, the Mohr-Coulomb rheology (MC) was developed and used to represent the behavior of geophysical mass flows \citep{SavageHutter1989}.

Shear and normal stress are assumed to obey Coulomb friction equation, both within the flow and at its boundaries. In other words,
\begin{equation}
\tau = \sigma \tan \phi,
\end{equation}
where $\tau$ and $\sigma$ are respectively the shear and normal stresses on failure surfaces, and $\phi$ is a friction angle. This relationship does not depend on the flow speed.

We can summarize the MC rheology assumptions as:
\begin{itemize}
\item \textit{Basal Friction} based on a constant friction angle.

\item \textit{Internal Friction} based on a constant friction angle.

\item \textit{Earth pressure coefficient} formula depends on the Mohr circle (implicitly depends on the friction angles).

\item Velocity based \textit{curvature effects} are included into the equations.
\end{itemize}

Under the assumption of symmetry of the stress tensor with respect to the \textit{z} axis, the earth pressure coefficient $k=k_{ap}$ can take on only one of three values $\{ 0, \pm 1\}$. The material yield criterion is represented by the two straight lines at angles $\pm \phi$ (the internal friction angle) relative to horizontal direction. Similarly, the normal and shear stress at the bed are represented by the line $\tau=-\sigma \tan(\delta)$ where $\delta$ is the bed friction angle.

\paragraph{MC equations} As a result, we can write down the source terms of the Eqs. (\ref{eq:D_A}):
\begin{eqnarray}\label{S_terms_MC}
S_x =& g_x h  - \frac{\bar{u}}{\| \underset{^\sim}{\bar{\textbf u}} \|} \left[h\left(g_z+\frac{\bar{u}^2}{r_x}\right)\tan(\phi_{bed})\right] - h k_{ap} \ {\rm sgn}\left(\frac{\partial \bar{u}}{\partial y}\right) \frac{\partial (g_z h)}{\partial y} \sin(\phi_{int}) \nonumber \\
 S_y =& g_y h  - \frac{\bar{v}}{\| \underset{^\sim}{\bar{\textbf u}} \|} \left[h\left(g_z +\frac{\bar{v}^2}{r_y}\right)\tan(\phi_{bed})\right] - h k_{ap} \ {\rm sgn}\left({\frac{\partial \bar{v}}{\partial x}}\right) \frac{\partial (g_z h)}{\partial x} \sin(\phi_{int})
\end{eqnarray}
Where, $\underset{^\sim}{\bar{\textbf u}} = (\bar{u} , \bar{v})$, is the depth-averaged velocity vector, $r_x$ and $r_y$ denote the radii of curvature
of the local basal surface. The inverse of the radii of curvature is usually approximated with the partial derivatives of the basal slope, e.g., $1/r_x = \partial \theta_x/\partial x$, where $\theta_x$ is the local bed slope.

In our study, sampled input parameters are $\phi_{bed}$, and $\Delta \phi:=\phi_{int}-\phi_{bed}$. In particular, the range of $\phi_{bed}$ depends on the case study, while $\Delta \phi \in [2^{\mathrm{\circ}}, 10^{\mathrm{\circ}}]$ \citep{Dalbey2008}.

\subsubsection{Pouliquen-Forterre}\label{PFM}
The scaling properties for granular flows down rough inclined planes led to the development of the Pouliquen-Forterre rheology (PF), assuming a variable frictional behavior as a function of Froude Number and flow depth \citep{Pouliquen1999, ForterrePouliquen2002, PouliquenForterre2002, ForterrePouliquen2003}.

PF rheology assumptions can be summarized as:
\begin{itemize}
\item \textit{Basal Friction} is based on an interpolation of two different friction angles, based on the flow regime and depth.

\item \textit{Internal Friction} is neglected.

\item \textit{Earth pressure coefficient} is equal to one.

\item Normal stress is modified by a \textit{pressure force} linked to the thickness gradient.

\item Velocity based \textit{curvature effects} are included into the equations.
\end{itemize}

Two critical slope inclination angles are defined as functions of the flow thickness, namely $\phi_{start}(h)$ and $\phi_{stop}(h)$. The function $\phi_{stop}(h)$ gives the slope angle at which a steady uniform flow leaves a deposit of thickness $h$, while $\phi_{start}(h)$ is the angle at which a layer of thickness $h$ is mobilized. They define two different basal friction coefficients.
\begin{eqnarray}
\mu_{start}(h)=\tan(\phi_{start}(h))\\
\mu_{stop}(h)=\tan(\phi_{stop}(h))
\end{eqnarray}

An empirical friction law $\mu_{b}(\|\underset{^\sim}{\bar{\textbf{u}}} \| , h)$ is then defined in the whole range of velocity and thickness. The expression changes depending on two flow regimes, according to a parameter $\beta$ and the Froude number $Fr=\| \underset{^\sim}{\bar{\textbf{u}}} \| / \ \sqrt{h g_{z}}$.

\paragraph{Dynamic friction regime - $Fr \ge \beta$}
\begin{equation}\label{mu_beta1}
\mu(h, Fr)=\mu_{stop}(h \beta / Fr)
\end{equation}

\paragraph{Intermediate friction regime - $0 \le Fr < \beta$}
\begin{equation}\label{mu_beta2}
\mu(h, Fr)=\left(\frac{Fr}{\beta}\right)^\gamma [\mu_{stop}(h)-\mu_{start}(h)] + \mu_{start}(h),
\end{equation}
where $\gamma$ is the power of extrapolation, assumed equal to $10^{-3}$ in the sequel \citep{PouliquenForterre2002}.

The functions $\mu_{stop}$ and $\mu_{start}$ are defined by:
\begin{equation}\label{mu-stop}
\mu_{stop}(h)=\tan\phi_{1} + \frac{\tan\phi_{2}-\tan\phi_{1}}{1+h/\it \mathcal{L}}
\end{equation}
and
\begin{equation}\label{mu-start}
\mu_{start}(h)=\tan\phi_{3} + \frac{\tan\phi_{2}-\tan\phi_{1}}{1+h/\it \mathcal{L}}
\end{equation}
The critical angles $\phi_{1}$, $\phi_{2}$ and $\phi_{3}$ and the parameters $\mathcal{L}, \beta$ are the parameters of the model.

In particular, $\mathcal{L}$ is the characteristic depth of the flow over which a transition between the angles $\phi_{1}$ to $\phi_{2}$ occurs, in the $\mu_{stop}$ formula. In practice, if $h\ll \mathcal L$, then $\mu_{stop}(h)\approx \tan\phi_{2}$, and if $h\gg \mathcal L$, then $\mu_{stop}(h)\approx\tan\phi_{1}$.

\paragraph{PF equations} The depth-averaged Eqs. (\ref{eq:D_A}) source terms thus take the following form:
\begin{eqnarray}\label{eq:S_terms_PF}
S_{x} &=&  g_{x} h -  \frac{\bar{u}}{\| \underset{^\sim}{\bar{\textbf{u}}} \|}\left[h \left(g_z+\frac{\bar{u}^2}{r_x}\right) \ \mu_{b}(\|\underset{^\sim}{\bar{\textbf{u}}} \| , h)\right] \ + g_{z}h\frac{\partial h}{\partial x} \nonumber \\
S_{y} &=&  g_{y} h - \frac{\bar{v}}{\| \underset{^\sim}{\bar{\textbf{u}}} \|}\left[h \left(g_z +\frac{\bar{v}^2}{r_y}\right) \ \mu_{b}(\|\underset{^\sim}{\bar{\textbf{u}}} \| , h)\right] \ + g_{z}h\frac{\partial h}{\partial y}
\end{eqnarray}

In our study, sampled input parameters are $\phi_1$, $\Delta \phi_{12}:=\phi_2-\phi_1$, and $\beta$. In particular, the range of $\phi_1$ depends on the case study, whereas $\Delta \phi_{12} \in [10^{\mathrm{\circ}}, 15^{\mathrm{\circ}}]$, and $\beta \in [0.1, 0.85]$. Moreover, $\phi_3=\phi_1+1^\mathrm{\circ}$, and $\mathcal{L}$ is equal to $1 dm$ and $1 mm$ in the two case studies, respectively \citep{PouliquenForterre2002,ForterrePouliquen2003}.

\subsubsection{Voellmy-Salm}\label{VSM}
The theoretical analysis of dense snow avalanches led to the VS rheology (VS) \citep{Voellmy1955, Salm1990, Salm1993, Bartelt1999}. Dense snow or debris avalanches consist of mobilized, rapidly flowing ice-snow mixed to debris-rock granules \citep{BarteltMcArdell2009}. The VS rheology assumes a velocity dependent resisting term in addition to the traditional basal friction, ideally capable of including an approximation of the turbulence-generated dissipation. Many experimental and theoretical studies were developed in this framework \citep{Gruber2007, Kern2009, Christen2010, Fischer2012}.

The following relation between shear and normal stresses holds:
\begin{equation}
\tau = \mu \sigma + \frac{\rho \| \underline{\textbf g} \|}{\xi} \| \underset{^\sim}{\bar{\textbf u}} \|^2,
\end{equation}
where, $\sigma$ denotes the normal stress at the bottom of the fluid layer and $\underline{\textbf g} = (g_{x} , g_{y} , g_{z})$ represents the gravity vector. The two parameters of the model are the bed friction coefficient $\mu$ and the velocity dependent friction coefficient $\xi$.

We can summarize VS rheology assumptions as:
\begin{itemize}
\item \textit{Basal Friction} is based on a constant coefficient, similarly to the MC rheology.

\item \textit{Internal Friction} is neglected.

\item \textit{Earth pressure coefficient} is equal to one.

\item Additional \textit{turbulent friction} is based on the local velocity by a quadratic expression.

\item Velocity based \textit{curvature effects} are included into the equations, following an different formulation from the previous models.
\end{itemize}

The effect of the topographic local curvatures is addressed with terms containing the local radii of curvature $r_x$ and $r_y$. In this case the expression is based on the speed instead of the scalar components of velocity \citep{PudasainiHutter2003,Fischer2012}.

\paragraph{VS equations} Therefore, the final source terms take the following form:
\begin{eqnarray}
\label{eq:S_terms_VS}
S_{x} &=&  g_{x} h - \frac{\bar{u}}{\| \underset{^\sim}{\bar{\textbf u}}\|} \ \left[ h \left(g_{z} + \frac{\| \underset{^\sim}{\bar{\textbf u}} \|^2}{r_{x}} \right)\mu+ \frac{\| \underset{^\sim}{\textbf g} \|}{\xi}\| \underset{^\sim}{\bar{\textbf u}} \|^2\right], \nonumber \\
S_{y} &=& g_{y} h - \frac{\bar{v}}{\| \underset{^\sim}{\bar{\textbf u}}\|} \ \left[ h \left(g_{z} + \frac{\| \underset{^\sim}{\bar{\textbf u}} \|^2}{r_{y}} \right)\mu+ \frac{\| \underset{^\sim}{\textbf g} \|}{\xi}\| \underset{^\sim}{\bar{\textbf u}} \|^2\right].
\end{eqnarray}

In our study, sampled input parameters are $\mu$, and $\xi$, on ranges depending on the case study. In particular, $\xi$ uniform sampling is accomplished in log-scale. In fact, values of $\xi$ between 250 and 4,000 $m/s^2$ have been described for snow avalanches \citep{Salm1993,Bartelt1999,Gruber2007}.

\subsection{Contributing variables}\label{sec:Fterms}
For analysis of modeling assumptions we need to record and classify the results of different modeling assumptions. In our case study, we focus on the right-hand side terms in the momentum equation and we call them RHS forces, or, more simply, the force terms. They are contributing variables since internal to the computation and rarely visible as a system output.
\begin{align}
\boldsymbol{RHS_1} = [g_x h,g_y h],
\end{align}
it is the gravitational force term, it has the same formulation in all models.

The expression of {\bf basal friction force} $\boldsymbol{RHS_2}$ depends on the model:
\begin{align}
\boldsymbol{RHS_2} =& -h g_z\tan(\phi_{bed})\left[\frac{\bar{u}}{\| \underset{^\sim}{\bar{\textbf u}} \|}, \frac{\bar{v}}{\| \underset{^\sim}{\bar{\textbf u}} \|} \right],\textmd{ in MC model.}\nonumber\\
\boldsymbol{RHS_2} =& - h g_z \ \mu_{b}(\|\underset{^\sim}{\bar{\textbf{u}}} \| , h)\left[\frac{\bar{u}}{\| \underset{^\sim}{\bar{\textbf{u}}} \|}, \frac{\bar{v}}{\| \underset{^\sim}{\bar{\textbf{u}}} \|}\right],\textmd{ in PF model.}\\
\boldsymbol{RHS_2} =& -h g_{z} \mu\left[\frac{\bar{u}}{\| \underset{^\sim}{\bar{\textbf u}}\|} , \frac{\bar{v}}{\| \underset{^\sim}{\bar{\textbf u}}\|}\right],\textmd{ in VS model.}\nonumber
\end{align}

The expression of the force related to the {\bf topography curvature}, $\boldsymbol{RHS_3}$, also depends on the model:
\begin{align}
\boldsymbol{RHS_3} =&-h \tan(\phi_{bed})\left[\frac{\bar{u}^3}{r_x\| \underset{^\sim}{\bar{\textbf{u}}} \|}, \frac{\bar{v}^3}{r_y\| \underset{^\sim}{\bar{\textbf{u}}} \|}\right],\textmd{ in MC model.}\nonumber\\
\boldsymbol{RHS_3} =& -h\ \mu_{b}(\|\underset{^\sim}{\bar{\textbf{u}}} \|,h)\left[\frac{\bar{u}^3}{r_x\| \underset{^\sim}{\bar{\textbf{u}}} \|}, \frac{\bar{v}^3}{r_y\| \underset{^\sim}{\bar{\textbf{u}}} \|}\right],\textmd{ in PF model.}\\
\boldsymbol{RHS_3} =& -h \mu\left[\frac{\bar{u}\| \underset{^\sim}{\bar{\textbf u}} \|}{r_{x}},\frac{\bar{v}\| \underset{^\sim}{\bar{\textbf u}} \|}{r_{y}}\right],\textmd{ in VS model.}\nonumber
\end{align}

All the three models have an additional force term, having a different expressions and different meaning in the three models:
\begin{align}
\boldsymbol{RHS_4} =&  - h k_{ap}\sin(\phi_{int})\left[ \ {\rm sgn}(\frac{\partial \bar{u}}{\partial y}) \frac{\partial (g_z h)}{\partial y},\ {\rm sgn}({\frac{\partial \bar{v}}{\partial x}}) \frac{\partial (g_z h)}{\partial x}\right],\textmd{ in MC model.}\nonumber\\
\boldsymbol{RHS_4} =& g_{z}h\left[\frac{\partial h}{\partial x}, \frac{\partial h}{\partial y}\right],\textmd{ in PF model.}\\
\boldsymbol{RHS_4} =& -\frac{\| \underset{^\sim}{\textbf g} \|}{\xi}\| \underset{^\sim}{\bar{\textbf u}} \|^2\left[\frac{\bar{u}}{\| \underset{^\sim}{\bar{\textbf u}}\|} \ ,\frac{\bar{v}}{\| \underset{^\sim}{\bar{\textbf u}}\|}\right],\textmd{ in VS model.}\nonumber
\end{align}
These contributing variables can be analyzed locally and globally for discriminating among the different modeling assumption.

Finally, we also study the spatial integrals defined by $F(t)=\int_{\mathbb R^k}f(\underline{\textbf x},t) d\underline{\textbf x}$, where $d\underline{\textbf x}$ is the area of the mesh elements. This provides a global view of the results and is complementary to the observations taken locally. For instance, by integrating the scalar product of source terms in the momentum balance and velocity we can compare the relative importance of modeling assumptions when we seek accuracy on global quantities.

\section{Small scale flow on inclined plane and flat runway}\label{sec:QoIs}
Our first case study assumes very simple boundary conditions, and corresponds to a laboratory experiment fully described in \citep{Webb2004, Bursik2005, WebbBursik2016}. It is a classical flow down an inclined plane set-up, including a change in slope to an horizontal plane (Fig. \ref{fig:Ramp-first}). Modeling flow of granular material down an inclined plane was explored in detail by several studies, both theoretically and experimentally (e.g. \cite{RuyerQuil2000, Silbert2001, Pitman2003b}).

In our setting, four locations are selected among the center line of the flow to accomplish local testing. These are: the initial pile location $L_1=(-0.7,0)$ m, the middle of the inclined plane $L_2=(-0.35,0)$ m, the change in slope $L_3=(0,0)$ m, the middle of the flat runway $L_4=(0.15,0)$ m.

\subsection{Preliminary consistency testing of the input ranges}\label{consistency}
Addressing a similar case study, \citep{Dalbey2008} assumed $\phi_{bed}=[15^\mathrm{\circ}, 30^\mathrm{\circ}]$, while \citep{WebbBursik2016} performed a series of laboratory experiments and found $\phi_{bed}=[18.2^\mathrm{\circ}, 34.4^\mathrm{\circ}]$. We relied on those published parameter choices to select a comprehensive parameter range. Figure \ref{fig:Ramp-first}b displays the screenshots of flow height observed in the extreme cases tested. The Digital elevation Map (DEM) has a 1mm cell size. Simulation options are - max\_time = 2 s, height/radius = 1.34, length\_scale = 1 m, number\_of\_cells\_across\_axis = 10, order = first, geoflow\_tiny = 1e-4 \citep{Patra2005,Aghakhani2016}. Initial pile geometry is cylindrical. We remark that small changes in the parameter ranges did not change significantly the results.

\begin{itemize}
\item \textbf{Material Volume:} $[449.0 \ ,\ 607.0] \ cm^3$, i.e. average of $528.0 \ cm^3$ and uncertainty of $\pm15\%$.
\item \textbf{Rheology models' parameters:}
\par\noindent \textbf{MC} - $\phi_{bed} \in [18^{\mathrm{\circ}}, 30^{\mathrm{\circ}}]$.

\vskip.1cm\noindent \textbf{PF} - $\phi_1 \in [10^{\mathrm{\circ}}, 22^{\mathrm{\circ}}]$.

\vskip.1cm\noindent \textbf{VS} - $\mu \in [0.22, 0.45]$, $\quad \log(\xi) \in [3, 4]$.
\end{itemize}

Fig \ref{fig:Ramp-first}b shows that even if maximum and minimum runout are both matching, the shape and lateral extent of the flow are different between the three models. In particular, MC model can produce the largest lateral extent, and the flow runout displays a larger lateral extent in PF. VS model displays an accentuated bow-like shape - the lateral wings remain behind the central section of the flow. This is due to the increased friction in the lateral margins.

\subsection{Observable outputs} \label{Obs1}
We express the flow height and acceleration as a function of time, measured in the four locations $L_1,\dots, L_4$ displayed in Fig. \ref{fig:Ramp-first}a. Uncertainty quantification (UQ) is performed, accordingly to the parameter ranges described above. We always show the mean values and the corresponding 5$^{\mathrm{th}}$ and 95$^{\mathrm{th}}$ percentile values, defining an uncertainty range.
\begin{figure}[H]
    \includegraphics[width=0.95\textwidth]{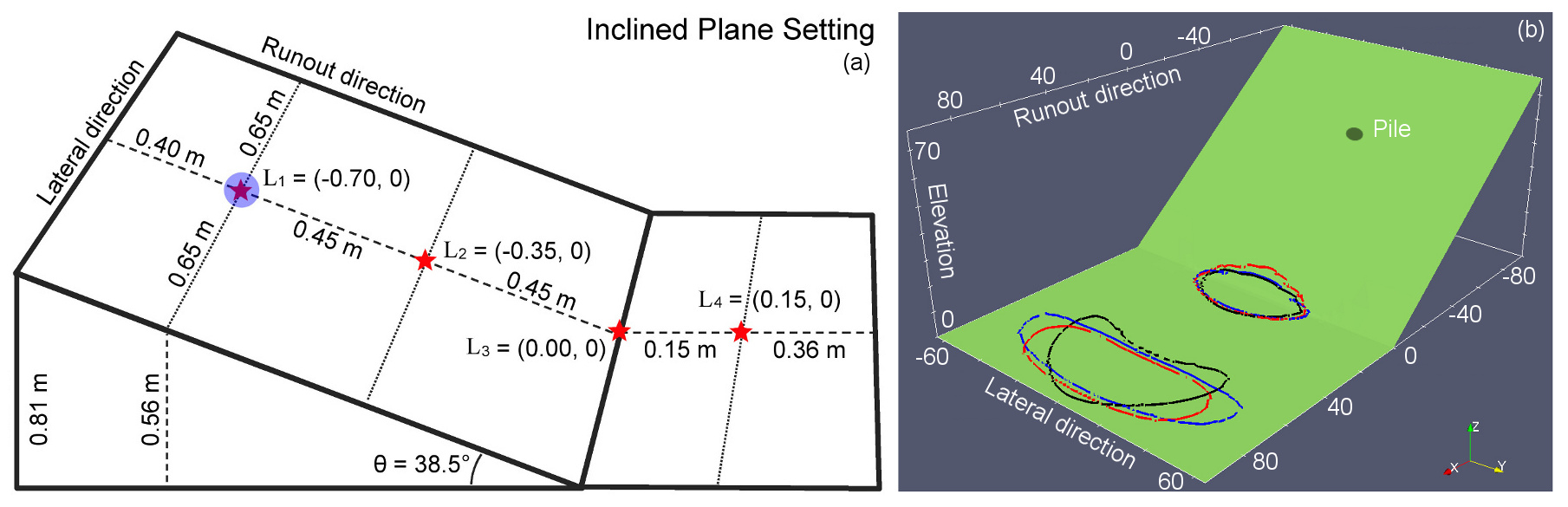}
    \centering
    \caption{(a) Inclined plane overview, including samples sites (red stars). Pile location is marked by a blue dot. (b) Contours of $h = 1.0$ mm at last simulated snapshot ($t = 1.5$ sec) for simulated flows with \emph{minimum runout} obtained from \emph{\textbf{min. volume -- max. resistance}}, and \emph{maximum runout} obtained from \emph{\textbf{max. volume -- min. resistance}}. {\color{red} \textbf{---}} : MC, {\color{blue} \textbf{---}} : PF, \textbf{---} : VS.} \label{fig:Ramp-first}
\end{figure}
\subsubsection{Flow height}
Figure \ref{fig:Ramp-H} displays the flow height, $h(L,t)$, at the points $(L_i)_{i=1,\dots,4}$. Given a particular type of flow and collected data we can clearly distinguish model skill in capturing not only that flow, but also other possible flows. Past work \citep{Webb2004, Patra2005} allowed us to conclude that MC rheology is adequate for modeling simple dry granular flows. We must note the effect of neglecting the flow when its height is $<1$ mm, which is at the scale of the smallest granular size \citep{Aghakhani2016}. In fact, continuum assumption would not be valid below this scale, and the 5$^{\mathrm{th}}$ and 95$^{\mathrm{th}}$ percentile plots are vertically cut to zero when they decrease over that threshold. The mean plot is not cut to zero but it is dulled by this cutoff.

In plot \ref{fig:Ramp-H}a, related to point $L_1$ placed on the initial pile, the values of $\sim 6\pm 1$ cm are equal and express the assigned pile height. The flow height decreases slightly faster in PF model, and slower in MC, compared to VS. Differences are more significant in plot \ref{fig:Ramp-H}b, related to point $L_2$, placed in the middle of the slope. Maximum flow height on average is greater in VS, $4.1\pm 0.2$ mm, but more uncertain in MC, $3.9\pm 0.4$ mm, and generally smaller in PF model, $3.0\pm 0.1$ mm. After the peak, PF decreases significantly slower than the other models. These height values are about 15 times smaller than initial pile height. None of the models leaves a significant material deposit in $L_1$ or $L_2$, and hence the 95$^{\mathrm{th}}$ percentile of the height is null at the ending-time. In contrast, a deposit is left at points $L_3$ and $L_4$, i.e. plot \ref{fig:Ramp-H}c placed at the change in slope, and plot \ref{fig:Ramp-H}d in the middle of the flat runway. At $L_3$ MC's deposit, $2$ mm with uncertainty [-2,+8] mm, is higher than the other models' deposits. The plot profile is bimodal, showing a first peak at $\sim 0.6$ s, and then a reduction until $1$ s, before the final accumulation. At $L_4$, deposit it is not significantly different between the three models. It measures $\sim 3$ mm on average, slightly more than this in VS, with uncertainty [-3,+7] mm.

In summary, MC is more distally stretched, but starts to deposit material earlier and closer to the initial pile compared to the other models. PF height is generally shorter, and is slightly earlier in its arrival at the sample points. These features are probably due to the correction term $g_z h \frac{\partial h}{\partial x}$ and $g_z h \frac{\partial h}{\partial y}$ which additionally pushes the material forward during the initial pile collapse. A linear cut in the flow height profile of PF is also observed when the flow thins on the slope. That is probably generated by the interpolation between the two basal friction angles as a function of flow height and speed. VS tends to be higher than the other models, if observed at the same instant, because of the reduced lateral spreading of material.
\begin{figure}[H]
         \centering
        \includegraphics[width=0.85\textwidth]{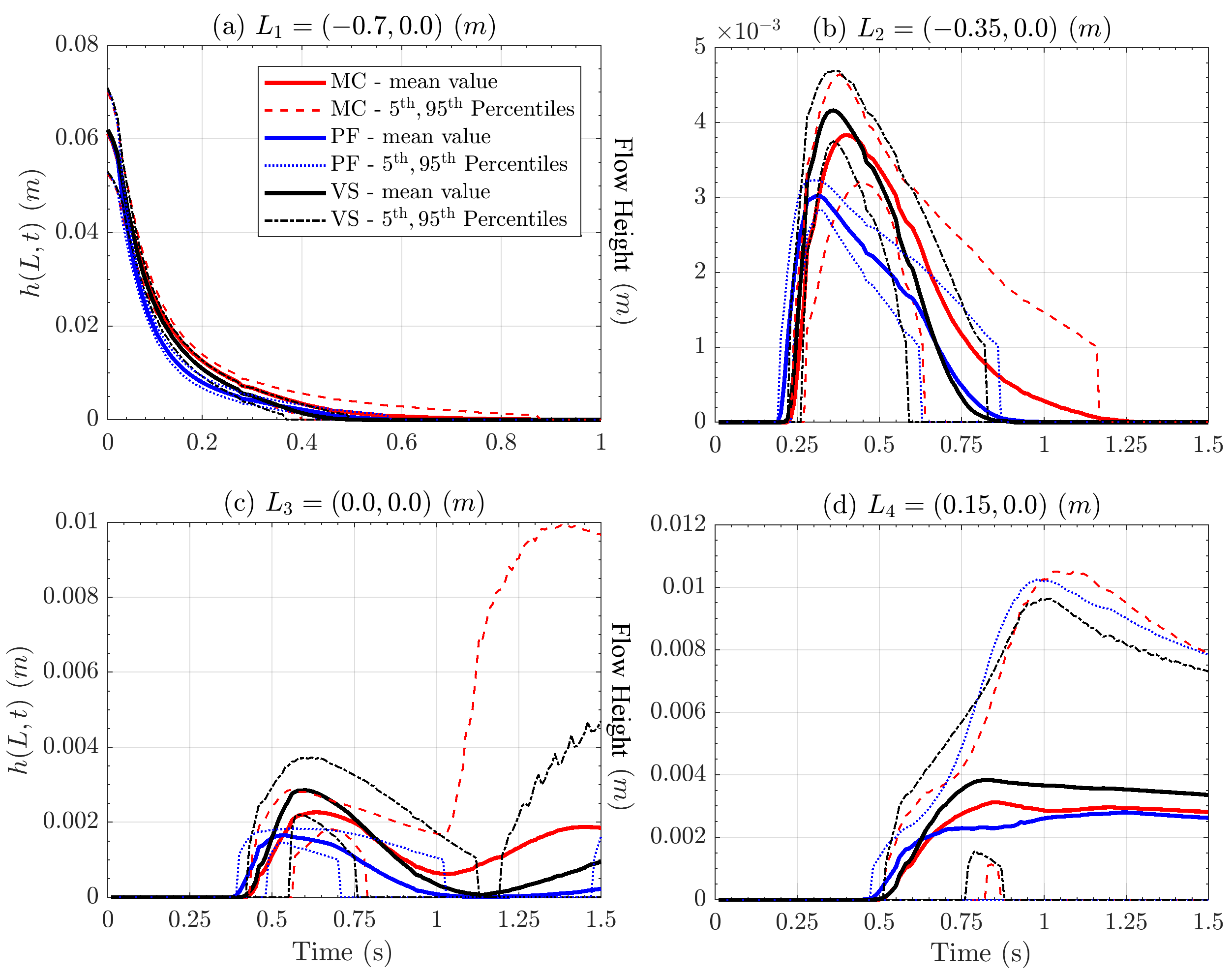}
        \caption{Flow height in four locations. Bold line is mean value, dashed/dotted lines are 5$^{\mathrm{th}}$ and 95$^{\mathrm{th}}$ percentile bounds. Different models are displayed with different colors. Plots are at different scale, for simplifying exposition.}
        \label{fig:Ramp-H}
\end{figure}
\subsubsection{Flow acceleration}
Figure \ref{fig:Ramp-AccL} shows the flow acceleration, $\Vert \underline{\mathbf{a}} \Vert(L,t)$, at the points $(L_i)_{i=1,\dots,4}$. Flow acceleration allows us to analyze the dynamics of the flow. We calculated it from the left-hand-side of the dynamical equation, but using the right-hand-side terms produces very similar results.

In plot \ref{fig:Ramp-AccL}a, related to point $L_1$, MC and VS show a plateau before $\sim 0.4$ sec, at $\sim 2.5 \ m/s^2$ and $\sim 3.5 \ m/s^2$, respectively, while PF linearly decreases between those same values. Instead in plot \ref{fig:Ramp-AccL}b, related to point $L_2$, MC and PF show a plateau, at $\sim 2.2$ m/s$^2$, while VS has a more bell-shaped profile. UQ tells us that PF is affected by a smaller uncertainty than the other models. In plot \ref{fig:Ramp-AccL}c, related to point $L_3$, all the models show a bimodal profile, with peaks at $\sim$ 0.5 sec and 0.8 sec. This is more accentuated in MC and VS, whereas the second peak is almost absent from PF's profile. At the first peak, acceleration values are significant, with average peaks in MC and PF both at $\sim 15 \ m/s^2$, and 95$^{\mathrm{th}}$ percentile plot reaching $\sim 50 \ m/s^2$ and $\sim 55 \ m/s^2$, respectively. VS shows about halved acceleration peak values. At the second peak, average acceleration values are similar in MC and VS, at $\sim 5 \ m/s^2$. In contrast, 95$^{\mathrm{th}}$ percentile plot is $> 50 \ m/s^2$ for MC, while $\sim 30 \ m/s^2$ in VS. In plot \ref{fig:Ramp-AccL}d, related to point $L_4$, the acceleration has a first peak at $\sim 4 \ m/s^2$, and a final asymptote at $\sim 2 \ m/s^2$ in MC and VS, $\sim 1 \ m/s^2$ for PF. These values indicate flow deceleration, and uncertainty is more relevant in MC and PF than in VS.
\begin{figure}[H]
         \centering
        \includegraphics[width=0.85\textwidth]{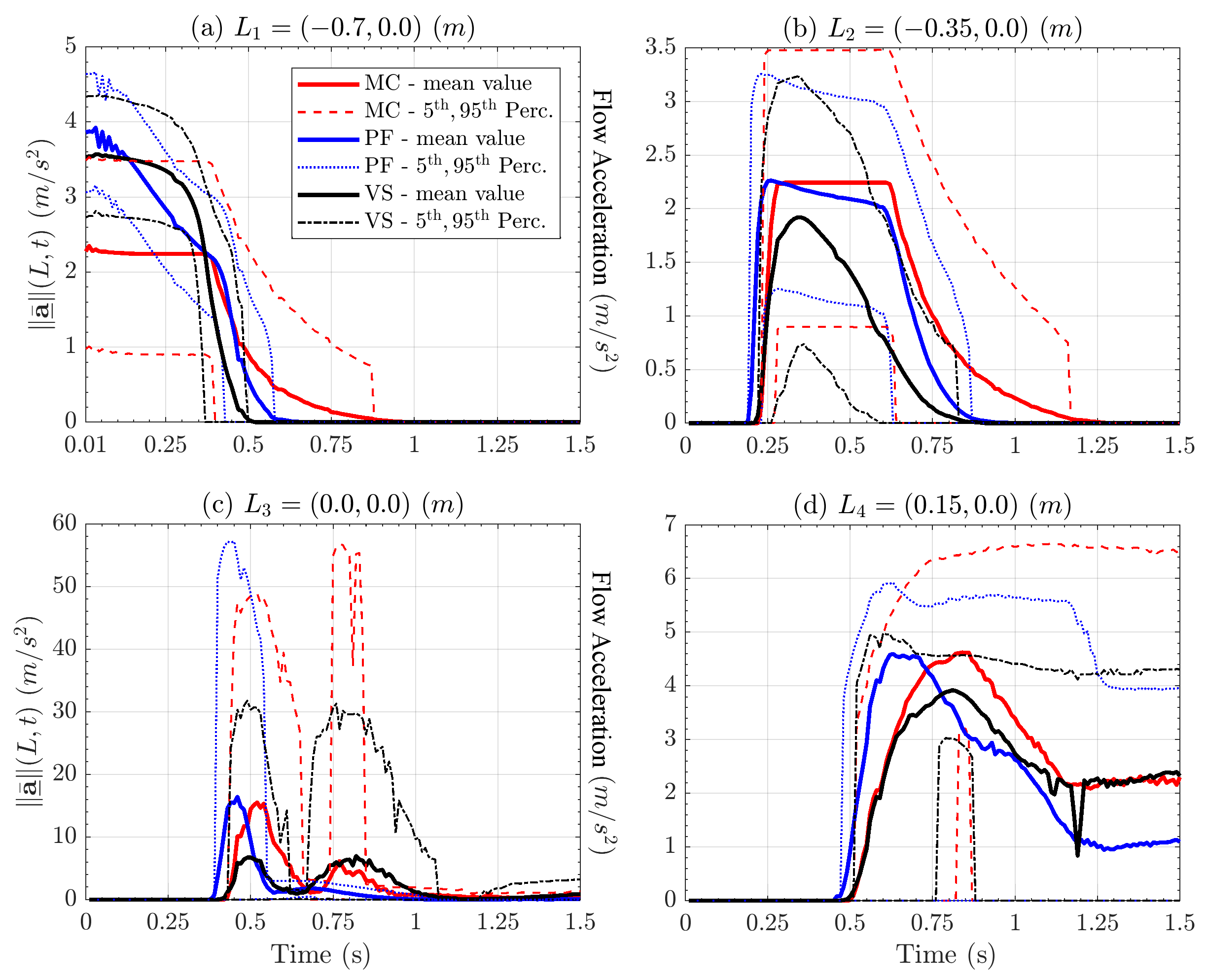}
        \caption{Flow acceleration modulus in four locations. Bold line is mean value, dashed lines are 5$^{\mathrm{th}}$ and 95$^{\mathrm{th}}$ percentile bounds. Different models are displayed with different colors. Plots are at different scale.}
        \label{fig:Ramp-AccL}
\end{figure}
In summary, the differences observed at the sample point on the inclined part are a consequence of the assumptions behind the models - double bed friction angle in PF, and speed dependent term in VS. At the slope change point, VS and MC display a bimodal profile in the acceleration. This is not a statistical effect, and it is also observed in single simulations. The first maximum is when the head of the flow hits the ground, while the second maximum is when the accumulating material in the tail arrives there. In VS the maxima are equal, because the tail is not laterally spread and hits the ground compactly. In contrast, PF does not show such a second peak, due to the accentuated lateral spreading in the tail.

\subsection{Statistical analysis of contributing variables}\label{Hq1}
Figure \ref{fig:Ramp-Pr_x} shows the dominance factors $(P_i)_{i=1,\dots,4}$, obtained from the modulus of the forces in the slope direction. Each dominance factor is the probability of a force term to be the greatest one, and hence belongs to $[0,1]$. The plots include also the probability of no-flow being observed at the considered point.

The plots \ref{fig:Ramp-Pr_x}a,b,c are related to point $L_1$, placed on the initial pile. Only the gravity $\boldsymbol{RHS_1}$ can be the dominant variable, and no-flow probability is $(1-P_1)$. Same thing in the plots \ref{fig:Ramp-Pr_x}d,e,f related to point $L_2$, placed in the middle of the slope. Then, plots \ref{fig:Ramp-Pr_x}g,h,i are related to point $L_3$, placed at the change in slope. In $L_3$, the curvature-related $\boldsymbol{RHS_3}$ can be the dominant term for a short time, with a peak probability of $\sim 30\%$. Plots \ref{fig:Ramp-Pr_x}j,k,l are related to point $L_4$, placed in the middle of the flat runway. In $L_4$ only the basal friction $\boldsymbol{RHS_2}$ can be the dominant term, except in PF where there is a $\sim 10\%$ chance that $\boldsymbol{RHS_4}$ is the dominant term at the ending-time.
\begin{figure}[H]
         \centering
        \includegraphics[width=0.95\textwidth]{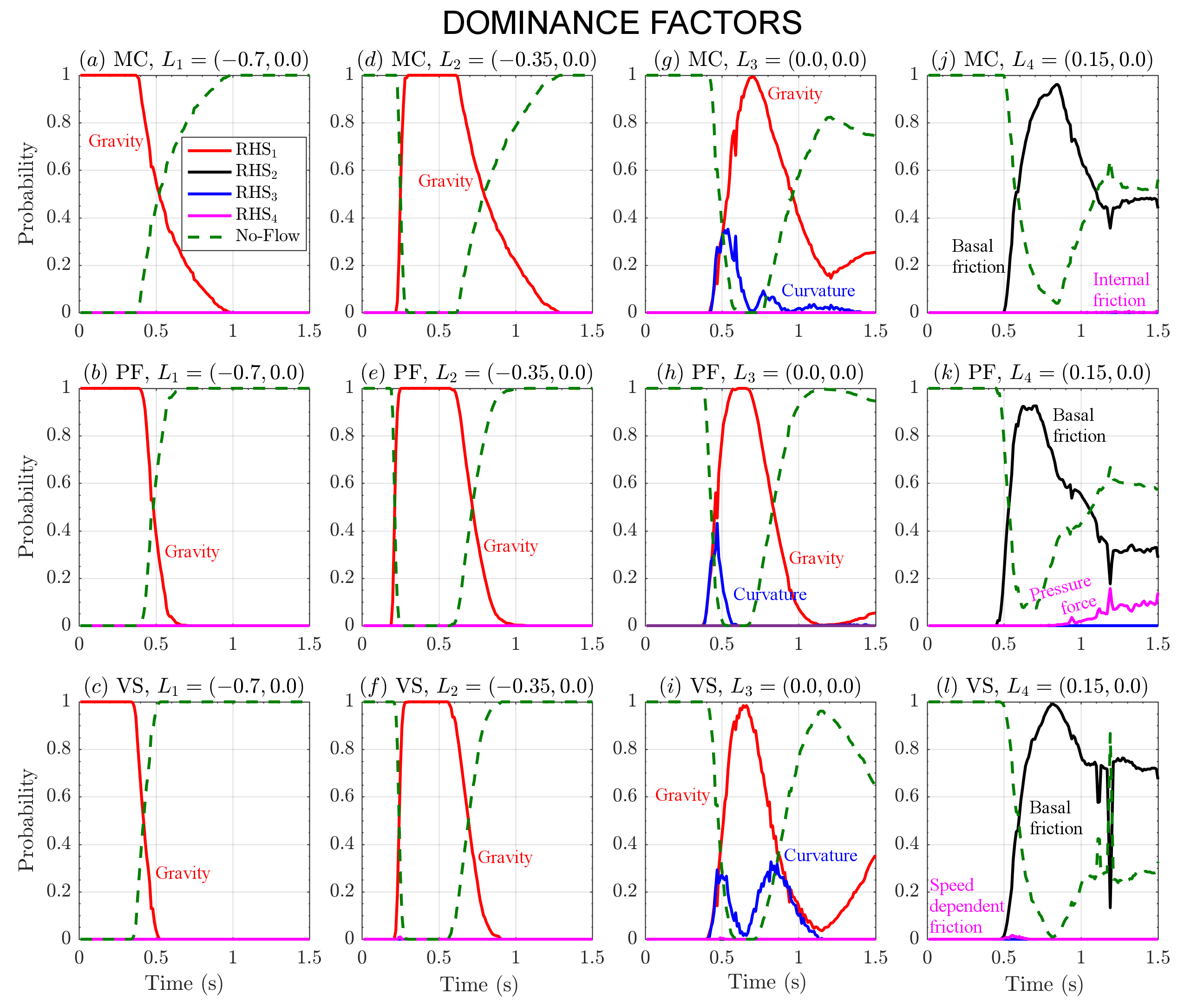}
        \caption{Dominance factors of the forces in the slope direction, in four locations. Different models are plotted separately: (a,d,g,j) assume MC; (b,e,h,k) assume PF; (c,f,i,l) assume VS. Different colors correspond to different force terms. No-flow probability is also displayed with a green dashed line.}
        \label{fig:Ramp-Pr_x}
\end{figure}
In this case there are minor differences in the dominance factors between the models. In particular, there is always a single dominant variable, and its profile is complementary with the no-flow probability. In the slope change point the differences between the models are more significant. Curvature term dominance probability is bimodal in MC and VS, and in VS the two peaks are equivalent. On the flat runway, in PF the pressure force can be the dominant variable with a small chance.

Dominance factors describe the main dynamics of the flow, but they are uninformative about the other variables. The expected contributions complete the statistical description of the contributing variables. They scale the force terms by the dominant dynamics, and represent the degree of relevance of the secondary assumptions with respect to the dominant one. They are averaged with respect to $P_M$ on the parameter range, and change as a function of time.

Figure \ref{fig:Ramp-Ci_x} shows $\mathbb E[C_i]_{i=1,\dots,4}$, for the three rheology models. $\forall i$, $C_i$ is related to the force term $\boldsymbol{RHS_i}$. We remark that in general, the second strongest force is never above 60\% of the dominant. The plots \ref{fig:Ramp-Ci_x}a,b,c are related again to point $L_1$. $C_1$ and $C_2$ give the major contributions, with a minor contribution from $C_4$ in VS. Contributions profiles are flat plateaus that start to wane after $0.4 s$. The plots \ref{fig:Ramp-Ci_x}d,e,f are related to point $L_2$. The major contributions are $C_1$ and $C_2$, and have a trapezoidal profile. In VS, $C_4$ resembles $C_2$, but it is bimodal instead than trapezoidal. The plots \ref{fig:Ramp-Ci_x}g,h,i are related to point $L_3$. $C_1$ and $C_2$ are still the largest, but their profiles are bell-shaped. In VS, $C_4$ is almost equal to $C_2$. In all the models, $C_3$ becomes significant, with a peak similar to $C_2$. It shows different profiles - triangular for MC and PF, bimodal for VS. In MC the decrease occurs in two stages. Due to the presence of deposit, all the contributions are small (particularly small in PF), but not zero at the ending time. The plots \ref{fig:Ramp-Ci_x}j,k,l are related to point $L_4$. Only $C_2$ has a major role, with a bell shaped profile faster to wax than to wane. Contribution $C_4$ has a minor role in VS and PF.
\begin{figure}[H]
         \centering
        \includegraphics[width=0.95\textwidth]{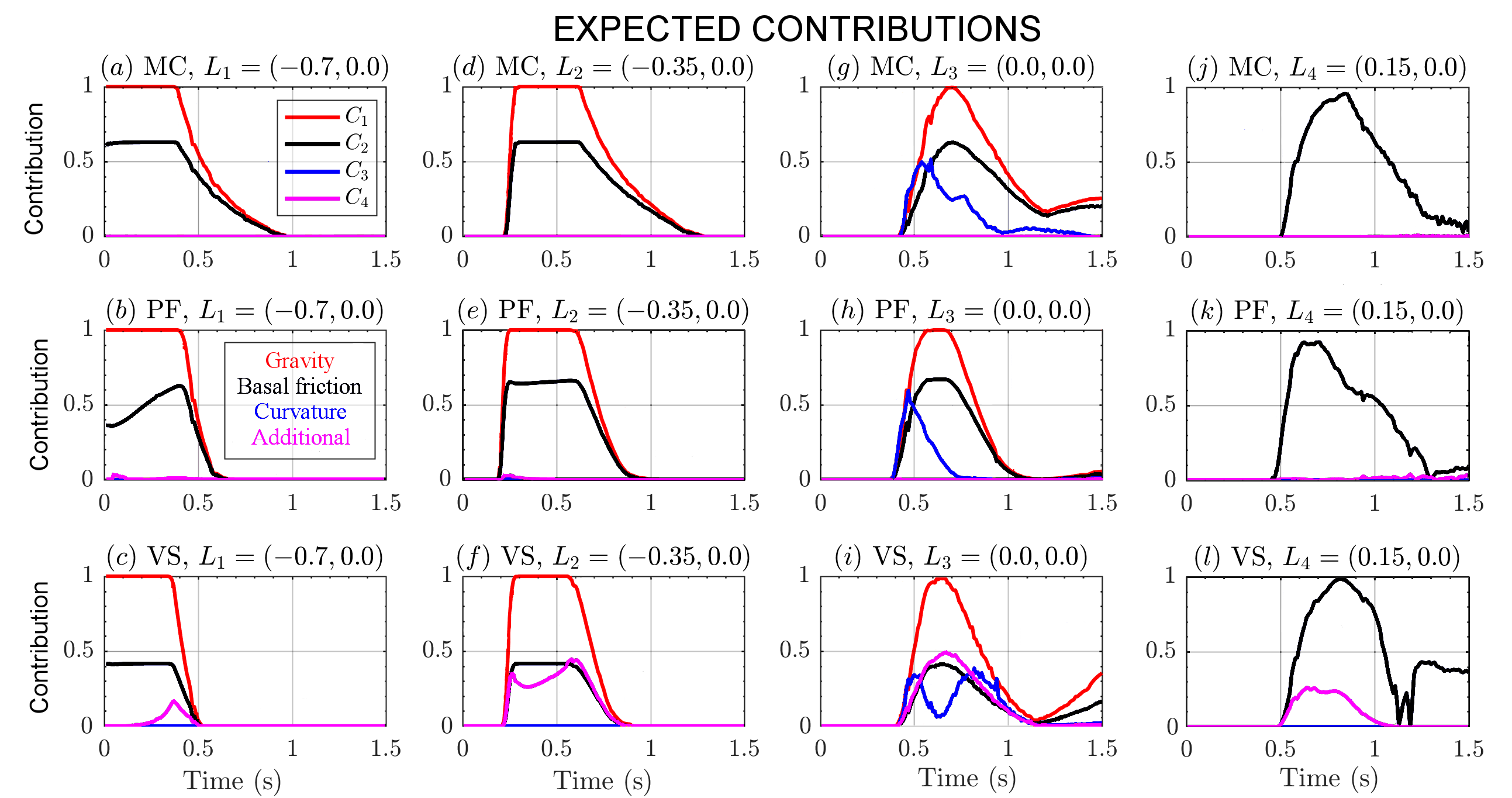}
        \caption{Expected contributions of forces in the slope direction, in four locations. Different models are plotted separately: (a,d,g,j) assume MC; (b,e,h,k) assume PF; (c,f,i,l) assume VS. Different colors correspond to different force terms.}
        \label{fig:Ramp-Ci_x}
\end{figure}
\subsection{Flow extent and spatial integrals}
Figure \ref{fig:Ramp-spatial} shows the volumetric average of speed and Froude Number. Estimates of local Froude Numbers are available in Supporting Information S1. Moreover, the figure shows the lateral extent and inundated area of flow, as a function of time. These global quantities have smoother plots than the local measurements describe above. However, most of the details observed in local measurements are not easy to discern anymore. In plot \ref{fig:Ramp-spatial}a the speed has a bell-shaped profile in all the models, with an average peak at $\sim 1.4 m/s$ and uncertainty range of $\pm 0.4 m/s$ for PF and VS. VS is slightly slower, reaching $\sim 1.3 m/s$ on average. MC shows a larger uncertainty range, of $\pm 0.6 m/s$. The maximum speed is reached first by VS and PF at $\sim 0.55 s$, and last by MC at $\sim 0.65 s$.

In plot \ref{fig:Ramp-spatial}b, also the Froude Number has a bell-shaped profile. $Fr$ peaks are temporally aligned with speed peaks, and are $\sim 10$ in VS, $\sim 11$ in MC, $\sim 13.5$ in PF, on average. Uncertainty range is about $\pm 4$ in all models. In plot \ref{fig:Ramp-spatial}c inundated area shows similar maximum values in PF and VS, at $\sim 0.6 m^2$ on average, and uncertainty of $\pm 0.15 m^2$. MC is lower, at $\sim 0.45 m^2$ on average, and less uncertain, $\pm 0.10 m^2$. VS does not decrease significantly after reaching the peak, whereas the other models contract their area to approximately half its maximum extent. In plot \ref{fig:Ramp-spatial}d the lateral extent starts equal to the pile diameter $\sim 15 cm$, and then rises in two stages in MC and PF. The second and greater rise starts at $\sim 0.6 s$, and corresponds with the time of arrival at the change in slope (see Fig. \ref{fig:Ramp-H}c). In contrast, VS rises without showing two phases. At $\sim 0.6 s$, average lateral extent is $\sim 50 cm$ in PF and VS, and $\sim 43 cm$ in MC. Uncertainty range is $\pm 7 cm$ for all models at that time. Final extent is $\sim 75 cm$ in PF, $\sim 65 cm$ in MC, $\sim 55 cm$ in VS. Uncertainty range is $\pm 5 cm$ in VS, but rises to $\pm 10 cm$ in MC and PF.
\begin{figure}[H]
        \centering
        \includegraphics[width=0.85\textwidth]{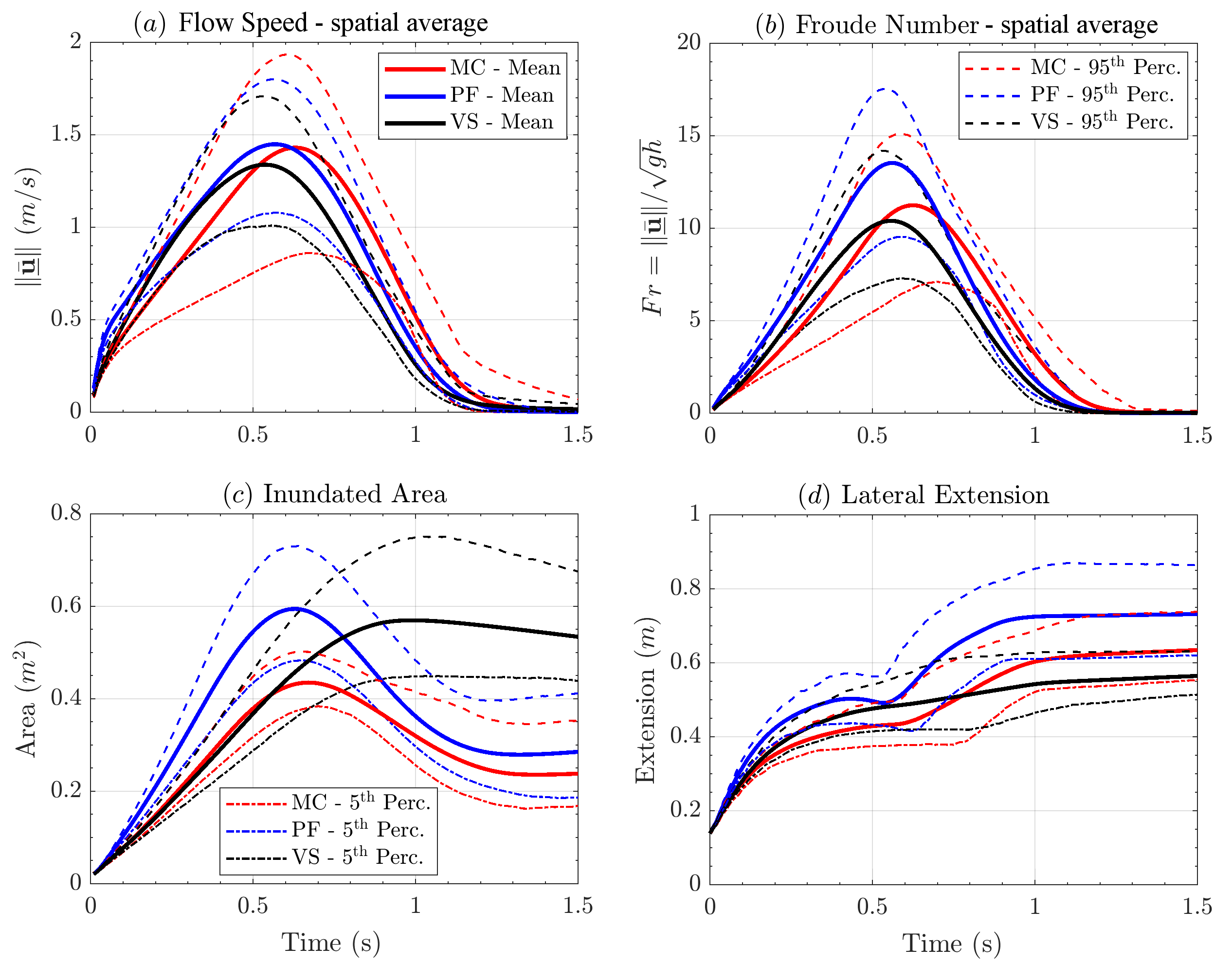}
        \caption{Comparison between spatial averages of $(a)$ flow speed, and $(b)$ Froude Number in addition to the flow $(c)$ lateral extent, and $(d)$ inundated area, as a function of time. Different models are displayed with different colors.}
        \label{fig:Ramp-spatial}
\end{figure}
In summary, the spatially averaged speed and Froude Number are significantly similar between the models. The differing features appear to be mostly localized in space. However, VS is significantly slower than the other models after the initial collapse. Moreover, it is the only model which presents a significant amount of long lasting and slowly moving material. Inundated area in PF has a greater maximum value, because of the accentuated lateral spread. In VS the inundated area almost does not decrease from its peak, because of the strictly increasing lateral spread. Vice versa, lateral spread in MC and PF has a temporary stop when the bulk of the flow hits the ground. This is a consequence of the interplay of accumulating material and the push of new material, which is stronger in the middle than in the lateral wings.

\subsection{Power integrals}
Figure \ref{fig:Ramp-Power-spatial} shows the spatial integral of powers (scalar product of force and velocity terms). The spatial integration is performed on half spatial domain, due to the symmetry with respect to the flow central axis. In particular, the power estimates assume a material density $\rho = 805 kg/m^3$. This is a constant scaling factor, and the plots are not further affected by its value.

Power terms have several features in common with the corresponding forces, and provide a decomposition of the acceleration sources. Main dissimilarity between forces and powers is that gravitational and basal friction powers have a profile starting from zero when the flow initiates, because the flow speed starts from zero. Corresponding plots of the force terms are included in Supporting Information S2. In plot \ref{fig:Ramp-Power-spatial}a the power of $\boldsymbol{RHS_1}$ represents the effect of the gravity in all the models. It starts from zero and rises up to $\sim 1.5 W$ at $\sim 0.55 s$, then decreases to zero after the material crosses the change in slope. Uncertainty range of $\pm 0.5 W$ affects the peak values. MC decreases slower than the other models, and has a more significant uncertainty after the change in slope. PF decreases faster.

In plot \ref{fig:Ramp-Power-spatial}b the power of  $\boldsymbol{RHS_2}$ represents the basal friction. It is negative and peaks to $\sim 1.1 \pm 0.2 W$ in MC, $\sim 1.0 \pm 0.2 W$ PF, $\sim 0.7 \pm 0.3 W$ in VS. A similar bell-shaped profile is shared by the three models, and basal friction power becomes negligible at the ending-time. In plot \ref{fig:Ramp-Power-spatial}c the power of $\boldsymbol{RHS_3}$ is related to the curvature effects, and it is not null only at the change in slope. It is always dissipative, i.e. opposed to flow velocity, because it is equivalent to the friction due to the additional weigh generated by centrifugal forces. It is weaker than $-0.1 W$ on average, ten times smaller than the previous powers, although MC lower percentile reaches $\sim -0.25 W$. VS displays a bimodal profile, with a second and weaker peak at $\sim 0.75 s$. In plot \ref{fig:Ramp-Power-spatial}d the power of $\boldsymbol{RHS_4}$ is related to the additional forces of the models, differently characterized. This term is really relevant in VS, although also in PF has a very short lasting positive peak up to $0.3 W$ before to become null at $\sim 0.1 s$. This power in VS is a speed dependent term, always dissipative. It is bell shaped and null before $\sim 0.1 s$ and after $\sim 1 s$. At the time of change in slope it is $\sim -0.7 W$, $\pm 0.3 W$.
\begin{figure}[H]
        \centering
        \includegraphics[width=0.90\textwidth]{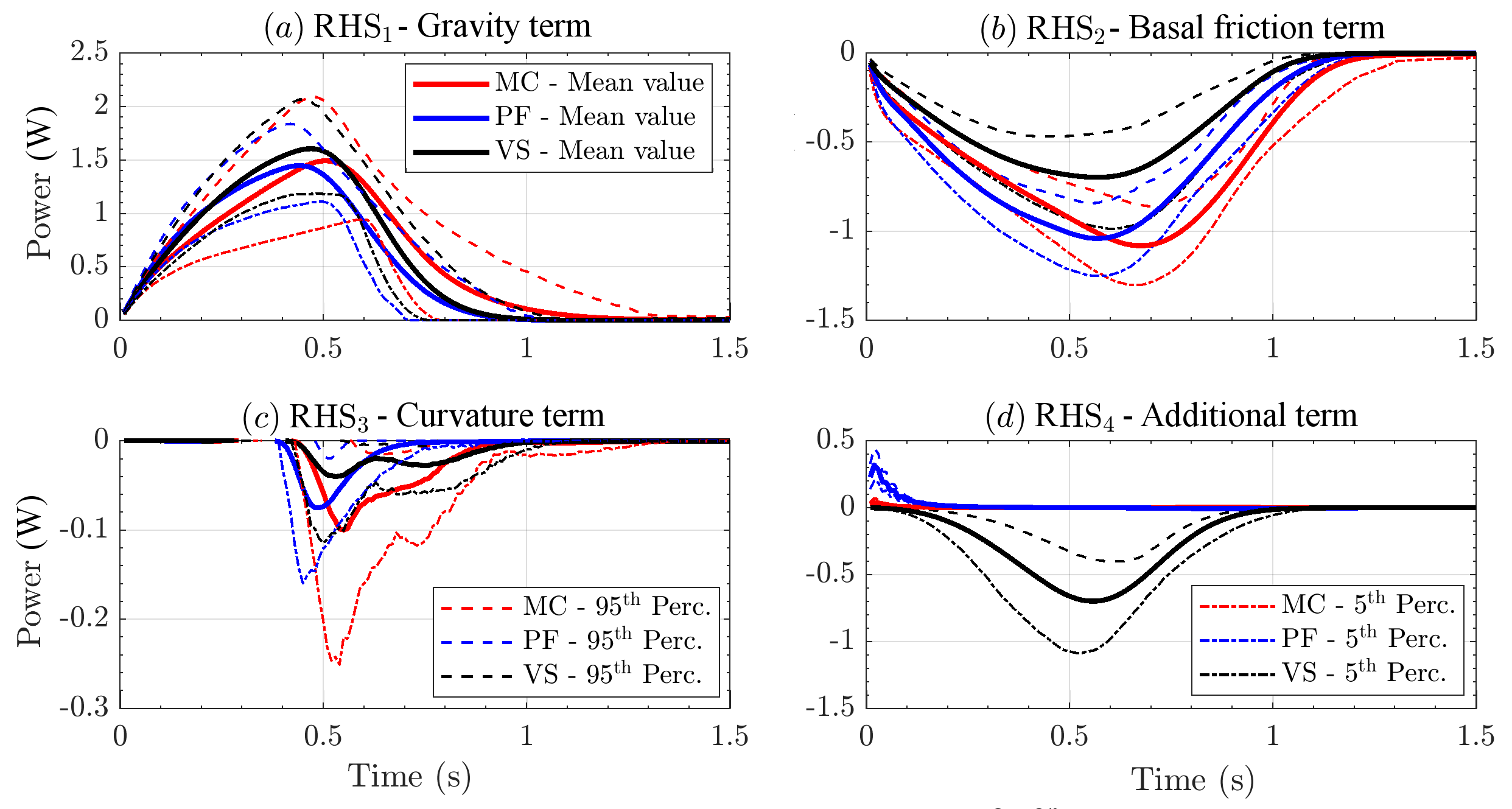}
        \caption{Spatial integral of the powers. Bold line is mean value, dashed lines are 5$^{\mathrm{th}}$ and 95$^{\mathrm{th}}$ percentile bounds. Different models are displayed with different colors.}
        \label{fig:Ramp-Power-spatial}
\end{figure}
In summary, this analysis produces a clear explanation of the source (gravity) and the dissipation of power. In particular, the dissipation lags the source (gravity) slightly. The curvature based term $\boldsymbol{RHS_3}$ has a minimal impact on VS. The differences between models are particularly relevant in term  $\boldsymbol{RHS_4}$. Speed dependent power in VS is at least one order of magnitude larger than the maximum values of the corresponding terms in MC and PF. Those are decreasing to zero after a short time from the initiation. Pressure force in PF is clearly positive in the speed direction, and hence contributing to push the flow ahead. The effects of internal friction in MC are almost negligible, and initially positive, then negative. This is motivated by an initial compression of the material during the pile collapse, followed by its stretching. It is worth remarking that $\boldsymbol{RHS_2}$ and $\boldsymbol{RHS_3}$ are both smaller in VS, due to the lower basal friction angles involved.

\section{Large scale flow on the SW slope of Volc{\'a}n de Colima}\label{QoI2}
Our second case study is a pyroclastic flow down the SW slope of Volc{\'a}n de Colima (MX) - an andesitic stratovolcano that rises to 3,860 m above sea level, situated in the western portion of the Trans-Mexican Volcanic Belt (Fig. \ref{fig:Colima-first}). Volc{\'a}n de Colima has historically been the most active volcano in M{\'e}xico \citep{DeLaCruzReina1993, Zobin2002, Gonzalez2002}. Pyroclastic flows generated by explosive eruptions and lava dome collapses of Volc{\'a}n de Colima are a well studied topic \citep{DelPozzo1995,Sheridan1995,Saucedo2002,Saucedo2004,Saucedo2005,Sarocchi2011,Capra2015}. The presence of a change in slope and multiple ravines characterize the SW slope of the volcano. Volc{\'a}n de Colima has been used as a case study in several research papers involving the Titan2D code \citep{Rupp2004, Rupp2006, Dalbey2008, Yu2009, Sulpizio2010, Capra2011, Aghakhani2016}. On July 10$^{\mathrm{th}}$-11$^{\mathrm{th}}$, 2015, the volcano underwent its most intense eruptive phase since its Subplinian-Plinian 1913 AD eruption \citep{Saucedo2010, Zobin2015, ReyesDaVilla2016, Capra2016, Macorps2018}.
\begin{figure}[H]
    \includegraphics[width=0.95\textwidth]{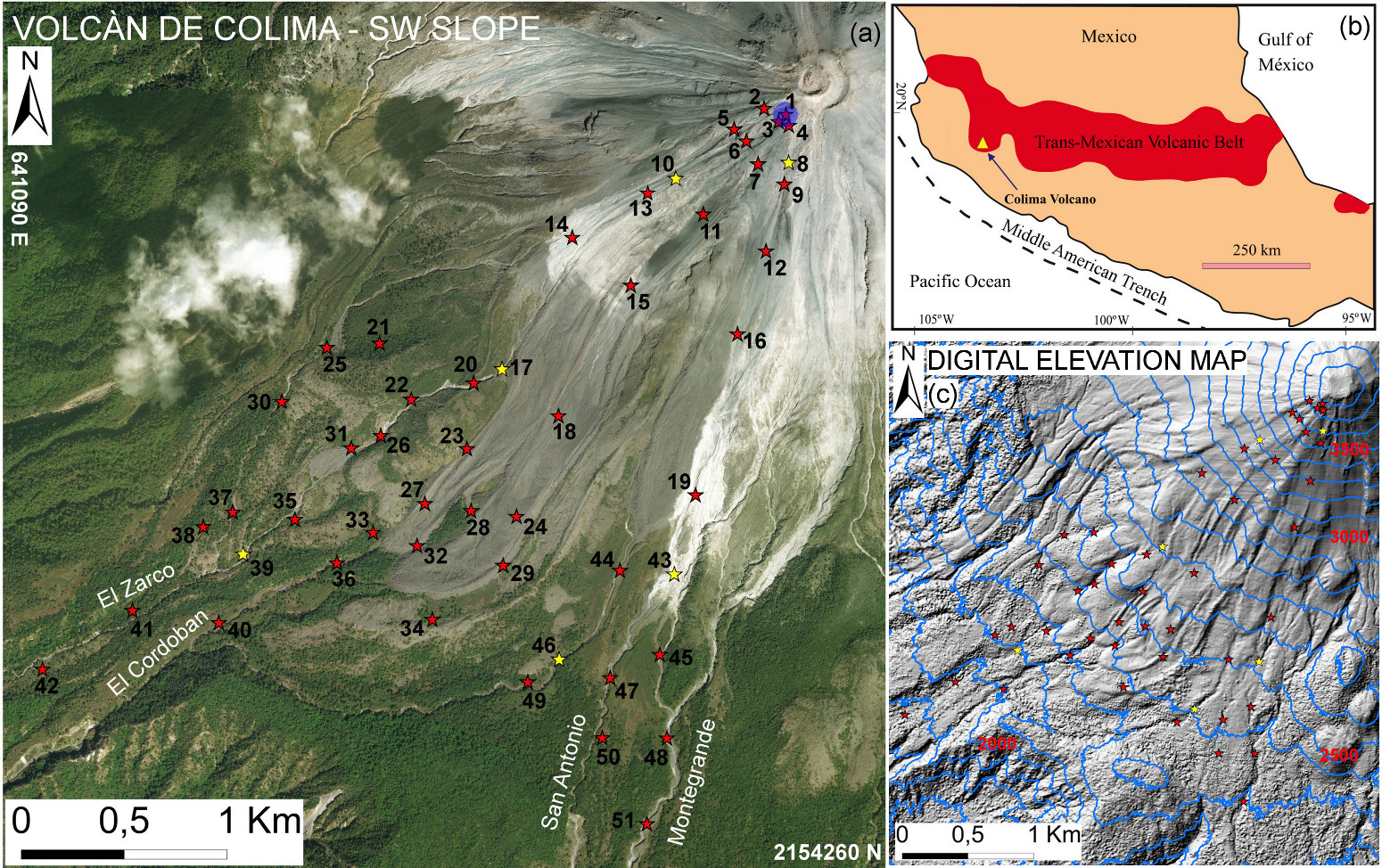}
    \centering
    \caption{(a) Volc{\'a}n de Colima (M{\'e}xico) overview, including 51 numbered locations (stars) and major ravines. Initial pile is marked by a blue dot. Coordinates are in UTM zone 13N. (b) Regional geology map. (c) Digital elevation map. Six preferred locations are colored in yellow. Elevation isolines are displayed in blue, elevation values in red.}
    \label{fig:Colima-first}
\end{figure}
We assume the flow to be generated by the gravitational collapse of a lava dome represented by a material pile placed close to the summit area - at 644956N, 2157970E UTM13 \citep{Rupp2006,Aghakhani2016}. A lava dome collapse occurs when there is a significant amount of recently-extruded highly-viscous lava piled up in an unstable configuration around a vent. Further extrusion and/or externals forces can cause the still hot dome of viscous lava to collapse, disintegrate, and avalanche downhill \citep{Bursik2005, Wolpert2016, Hyman2018}. The volcano produced several pyroclastic flows of this type, called Merapi style flows \citep{Macorps2018}. The hot, dense blocks in this ``block and ash'' flow (BAF) will typically range from centimeters to a few meters in size. Our computations were performed on a DEM of 5m-pixel resolution, obtained from Laser Imaging Detection and Ranging (LIDAR) data acquired in 2005 \citep{Davila2007, Sulpizio2010}. We placed 51 locations along the flow inundated area to accomplish local testing. After evaluating the results in all the locations, six of them are adopted as preferred locations, being representative of different flow regimes.

\subsection{Preliminary consistency testing of the input ranges}
In this same setting, \cite{Dalbey2008} assumed $\phi_{bed}=[15^\mathrm{\circ}, 35^\mathrm{\circ}]$, while \cite{Capra2011} adopted $\phi_{bed}=30^\mathrm{\circ}$.
\begin{figure}[H]
         \centering
        \includegraphics[width=0.85\textwidth]{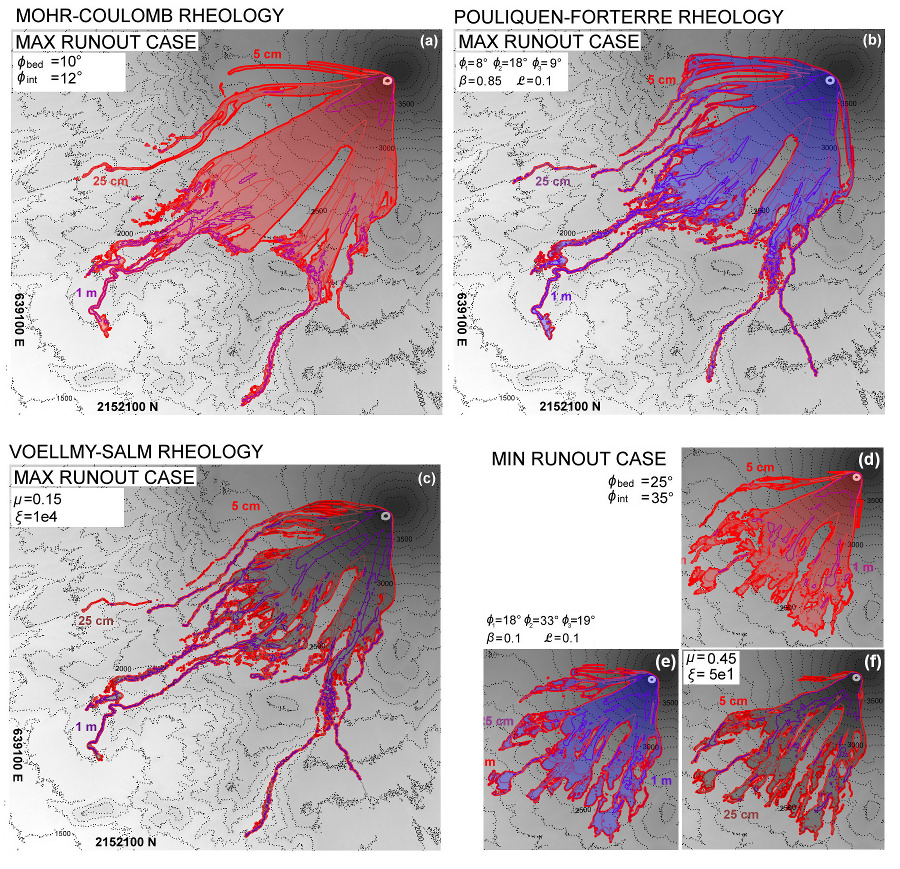}
        \caption{Volc\'an de Colima - comparison between \emph{max flow height} maps of simulated flow, assuming MC (a),(d), PF (b),(e), and VS (c),(f) models. Extreme cases - (a),(b),(c) \emph{\textbf{max. volume -- min. resistance}} and (d),(e),(f) \emph{\textbf{min. volume -- max. resistance}}.}
        \label{Colima-MaxMinExtents}
\end{figure}
Then, \cite{Spiller2014,Bayarri2015,Ogburn2016} found a statistical correlation between flow size and effective basal friction inferred from field observation of geophysical flows. A BAF at the scale of our simulations would possess $\phi_{bed}=[13^\mathrm{\circ}, 18^\mathrm{\circ}]$ according to their estimates. Small changes in the parameter ranges do not change significantly the results.

Figure \ref{Colima-MaxMinExtents} displays the maps of maximum flow height observed in the extreme cases tested. Simulation options are - max\_time = 7200 s (2 hours), height/radius = 0.55, length\_scale = 4e3 m, number\_of\_cells\_across\_axis = 50, order = first, geoflow\_tiny = 1e-4 \citep{Patra2005,Aghakhani2016}. Initial pile geometry is paraboloid.

Even if the maximum runout is matched between the models, they display significantly different macroscopic features. In particular, MC displays a further distal spread before entering the ravines, PF shows a larger angle of lateral spread at the initiation pile, and stops more gradually than MC with more complex inundated area boundary lines. VS is less laterally extended and the material reaches higher thickness. The flow generally looks significantly channelized, and displays several not-inundated spots due to minor topographical coul\'{e}es.

\begin{itemize}
\item \textbf{Material Volume:} $[2.08, 3.12] \times 10^5 \ m^3$, i.e. average of $2.6  \times 10^5 \ m^3$ and uncertainty of $\pm 20\%$.
\item \textbf{Rheology models' parameters:}
\par\noindent \textbf{MC} - $\phi_{bed} \in [10^{\mathrm{\circ}}, 25^{\mathrm{\circ}}]$.

\vskip.1cm\noindent \textbf{PF} - $\phi_1 \in [8^{\mathrm{\circ}}, 18^{\mathrm{\circ}}]$.

\vskip.1cm\noindent \textbf{VS} - $\mu \in [0.15, 0.45]$, $\quad \log(\xi) \in [1.7, 4]$.

\end{itemize}

\begin{figure}[H]
         \centering
        \includegraphics[width=0.90\textwidth]{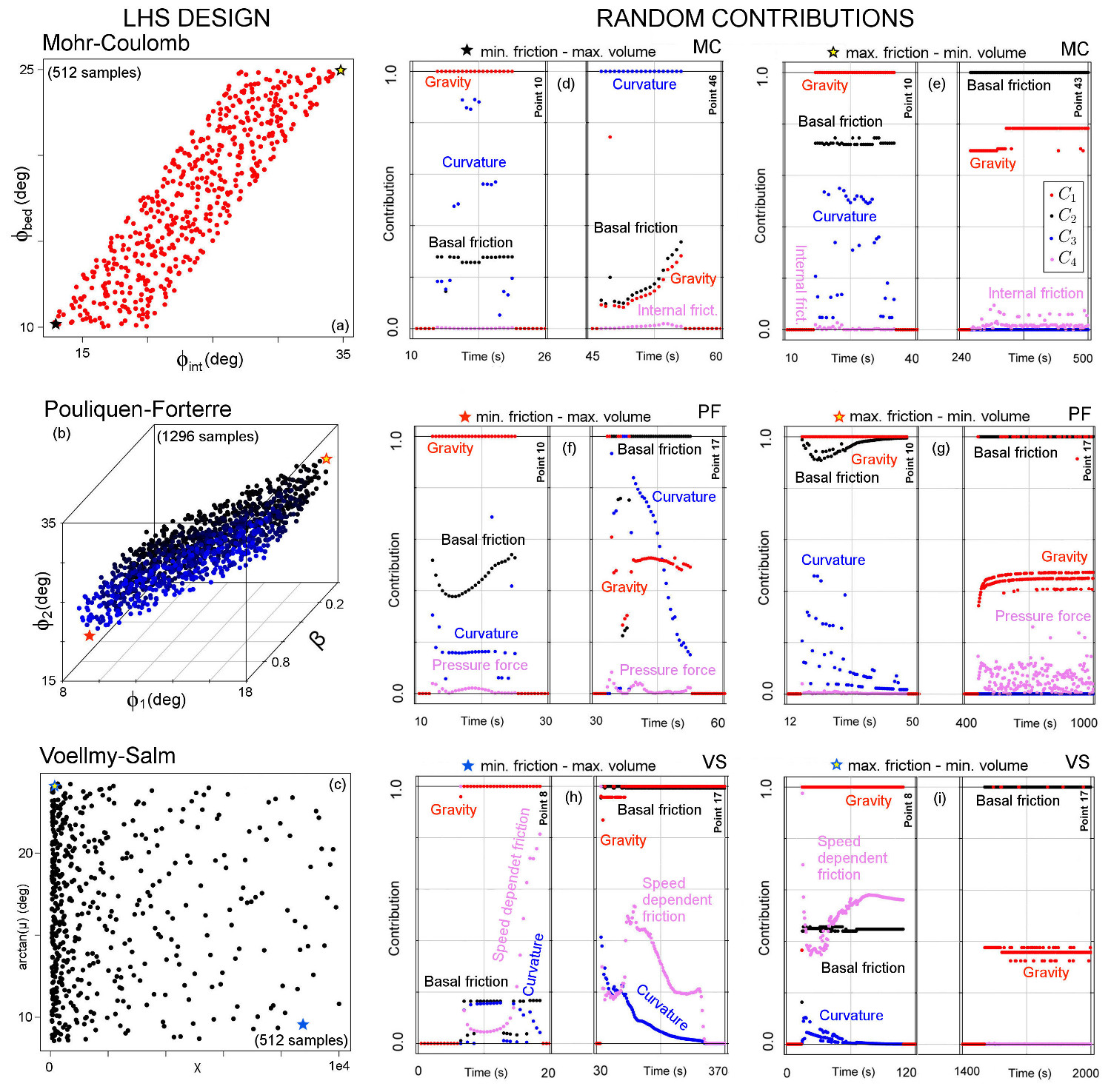}
        \caption{(a,b,c) Example of Latin Hypercube Sampling design, Volc\'an de Colima case study. Colored stars mark the values producing the minimum and maximum friction. The parameter values are projected with respect to their volume value. Plots of random contributions are included for (d,e) MC, (f,g) PF, and (h,i) VS model. Each plot includes two graphs. One refers to a proximal location to the initial pile (left), and the other to a more distal location (right). Point numbers refer to Fig.\ref{fig:Colima-extra}. Different colors correspond to different force terms.}
        \label{fig:Colima-CC1}
\end{figure}

\subsection{Exploring Flow Limits} \label{lhs_des_colima}
Figure \ref{fig:Colima-CC1}d,e,f,g,h,i show examples of the  contributions obtained assuming parameter values at the extremes of their range. The dominant variable is expressed by the dots on the top line, $C_i=1$. Data is inherently discontinuous due to the mesh modification, and it is reported with colored dots. If the mesh element which contains the considered spatial location changes, then the force term is calculated on a different region and suddenly changes too. We remark that this can also affect the dominant variable, and more than one random contribution can incorrectly appear to be at unity at the same time. However, it is evident that the dynamics and its temporal scale is evolving, and that the contributions can reveal a large amount of information about it.

We remark that, $\forall i$, the calculation of $\mathbb E[C_i]$ with respect to $P_M$ removes the effects of data discontinuity, and hence this is a fundamental step in our further analysis. We note that the above choices are easily changed, and if we are interested for instance in the performance of the models for very large or very small flows, a suitable volume range can be chosen and the procedure re-run.

\begin{figure}[H]
         \centering
        \includegraphics[width=0.95\textwidth]{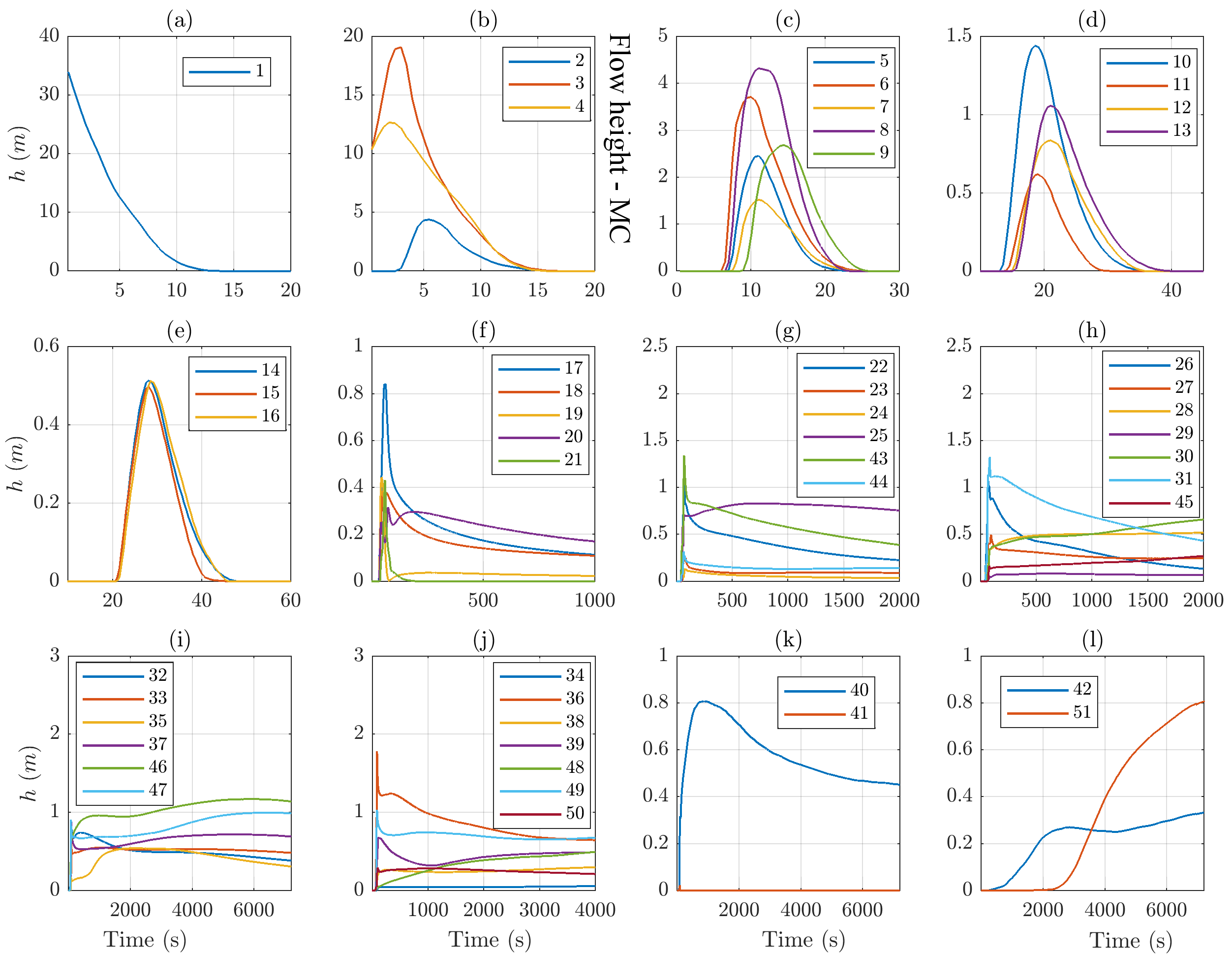}
        \caption{MC model, mean flow height $h(L,t)$ in 51 numbered locations (Fig. \ref{fig:Colima-first}). Different plots have different scales on either time and space axes.}
        \label{fig:BAF-H-MC}
\end{figure}

\subsection{Observable outputs}
The number of spatial locations is significantly high. We placed 51 points to span the entire inundated area, in search of different flow regimes, as displayed in Fig. \ref{fig:Colima-first}. These locations have an explorative purpose, whereas the six preferred locations will describe distinct flow regimes. We remark that all the distances reported in the following are measured in vertical projection, thus without considering the differences in elevation. Estimates of local Froude Numbers are available in Supporting Information S3 and S4.

Figure \ref{fig:BAF-H-MC} shows the mean flow height, $h(L,t)$, at the 51 spatial locations of interest, according to MC. In plot \ref{fig:BAF-H-MC}a, the only location is set on the center of the initial pile, and the profile is similar to what observed in point $L_1$ of the inclined plane case study, in Fig.\ref{fig:Ramp-H}a. In this case the height decreases from the initial value to zero in $\sim 15 s$. In plots \ref{fig:BAF-H-MC}b,c,d,e, the locations are are set at less than $\sim 1$ km radius from the initial pile. Their profiles are similar to point $L_2$ in Fig.\ref{fig:Ramp-H}b. The height profile is bell-shaped, starting from zero and then waning back to zero in $\sim 20$ s. All the dynamics occurs during the first minute. In plots \ref{fig:BAF-H-MC}f,g,h,i,j, points are set where the slope reduces, and the flow can channelize, and typically leaves a deposit. The distance from the initial pile is $\sim 2-3$ km. The profiles are sometimes similar to $L_3$ of Fig.\ref{fig:Ramp-H}c, other times to $L_4$ of Fig.\ref{fig:Ramp-H}d, in a few cases showing intermediate aspects. In general is either observed an initial short-lasting bulge followed by a slow decrease lasting for several minutes and asymptotically tending to a positive height, or a steady increase of material height tending to a positive height. In both cases it is sometimes observed a bimodal profile in the first 5 minutes. Finally, plots \ref{fig:BAF-H-MC}k,l focus on three points set at about the runout distance of the flow, in the most important ravines, at $\sim 4-5$ km from the initial pile. Profiles are similar to what observed in point $L_4$ of Fig.\ref{fig:Ramp-H}d.

\begin{figure}[H]
         \centering
        \includegraphics[width=0.95\textwidth]{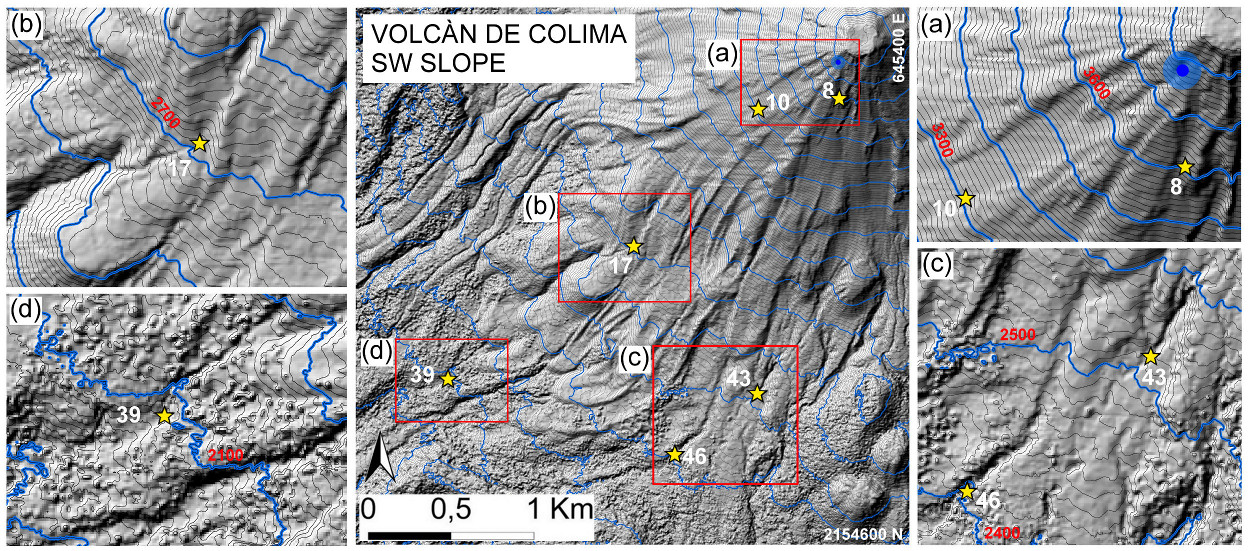}
        \caption{Volc{\'a}n de Colima (M{\'e}xico) overview, including six numbered locations (stars). In (a), (b), (c), (d) are enlarged the proximal topographic features to those locations. Initial pile is marked by a blue dot. Reported coordinates are in UTM zone 13N. Elevation isolines at every $10 m$ are displayed in black, at every $100m$ in bold blue. Elevation values in red.}
        \label{fig:Colima-extra}
\end{figure}

\subsubsection{Flow height in six locations}\label{Obs2}
We select six preferred locations, illustrative of a range of flow regimes. They are $[L_8, L_{10}, L_{17}, L_{39}, L_{43}, L_{46}]$, as displayed in Figure \ref{fig:Colima-extra}. The first two points, $L_8$ and $L_{10}$, are both proximal to the initiation pile. Points $L_{17}$ and $L_{43}$ are placed where the slope is reducing and the ravines are evident, and $L_{39}$ and $L_{46}$ are placed in the channels, further down-slope. In particular, $L_8$, $L_{43}$, and $L_{46}$ are at the western side of the inundated area, whereas $L_{10}$, $L_{17}$, and $L_{39}$ are at the eastern side. Estimates of flow acceleration are available in Supporting Information S5.

Figure \ref{fig:Colima-H} shows the flow height, $h(L,t)$, at the points $(L_i)_{i=8,10,17,39,43,46}$. Distances from the initial pile are in vertical projection. In plots \ref{fig:Colima-H}a,b, we show the flow height in points $L_8$ and $L_{10}$, $\sim 200 m$ and $\sim 500 m$ from the initial pile, respectively. Models MC and PF display similar profiles, positive for less than $15 s$ and bell-shaped. VS requires a significantly longer time to decrease, particularly in point $L_{10}$, where the average flow height is still positive after $\sim 200 s$. Peak average values in $L_8$ are $3.4 m$ in $PF$, $4.3 m$ in MC, $4.7 m$ in VS. Uncertainty is about $\pm 2 m$, halved on the lower side in $MC$, and $PF$. In $L_{10}$, models MC and PF are very similar, with peak height at $1.4 m$ and uncertainty $\pm 0.5 m$. Model VS, in contrast, has a maximum height of $1.1 m$ lasting for $50s$, and 95$^{th}$ percentile reaching $3.7 m$.
\begin{figure}[H]
         \centering
        \includegraphics[width=0.9\textwidth]{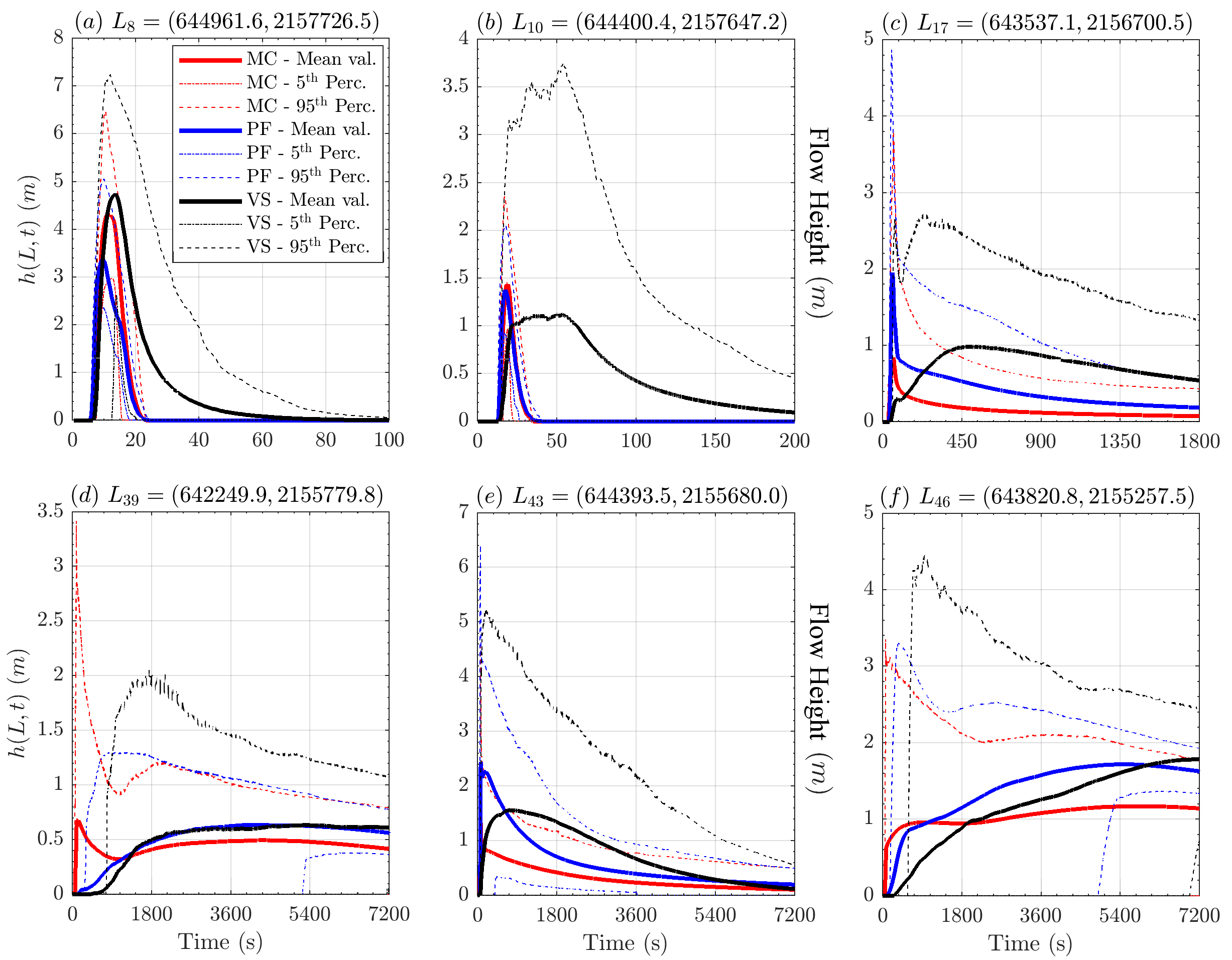}
        \caption{Flow height in six locations. Bold line is mean value, dashed/dotted lines are 5$^{\mathrm{th}}$ and 95$^{\mathrm{th}}$ percentile bounds. Different models are displayed with different colors. Plots are at different scale, for simplifying lecture.}
        \label{fig:Colima-H}
\end{figure}
In plots \ref{fig:Colima-H}c,e, we show the flow height in points $L_{17}$ and $L_{43}$, both at $\sim 2 km$ from the initial pile. All the models show a fast spike during the first minute, followed by a slow decrease. There is still material after $1800 s$. VS has a secondary rise peaking at $\sim 450 s$, which is not observed in the other models. This produces higher values for the most of the temporal duration, but similar deposit thickness after more than 1 hour. Maximum values are $1 m$ for MC, $2 m$ for PF, and $1.5 m$ for VS, in both locations. The 5$^{th}$ percentile is zero in all the models, meaning that the parameter range does not always allow the flow to reach these locations. The 95$^{th}$ percentile is above $5 m$, except for VS in point $L_{17}$. In plots \ref{fig:Colima-H}d,f, we show the flow height in points $L_{39}$ and $L_{46}$, both placed at more than $3 km$ from the initial pile. The three models all show a monotone profile except for  MC in point $L_{39}$, which instead displays an initial spike and a decrease before to rise again. A similar thing is observed in the 95$^{th}$ percentiles of all the models. It is significant that the 5$^{th}$ percentile of PF becomes positive after $\sim 5400 s$, meaning that almost surely the flow has reached that location. Deposit thickness in point $L_{39}$ is $\sim 0.5 m$ for all the models, whereas in point $L_{46}$ it is $1.7 m$ in VS, $1.6 m$ in PF, and $1.2 m$ in MC.

This analysis allows us to compare the local flow regimes with what is observed in the four sample points of the small scale case study. There is a new feature which was not present in the small scale flow - VS is temporally stretched, and material arrives later and stays longer in all the sample points. This is a consequence of the speed dependent term reducing flow velocity.

\subsection{Statistical analysis of contributing variables}\label{Hq2}
\subsubsection{Three locations proximal to the initial pile}
Figure \ref{fig:Colima-Pr1} shows the dominance factors $(P_i)_{i=1,\dots,4}$ of the RHS terms modulus, in the three proximal points $L_{8}$, $L_{10}$, and $L_{17}$, all closer than 1 km to the initial pile.
\begin{figure}[H]
         \centering
        \includegraphics[width=0.95\textwidth]{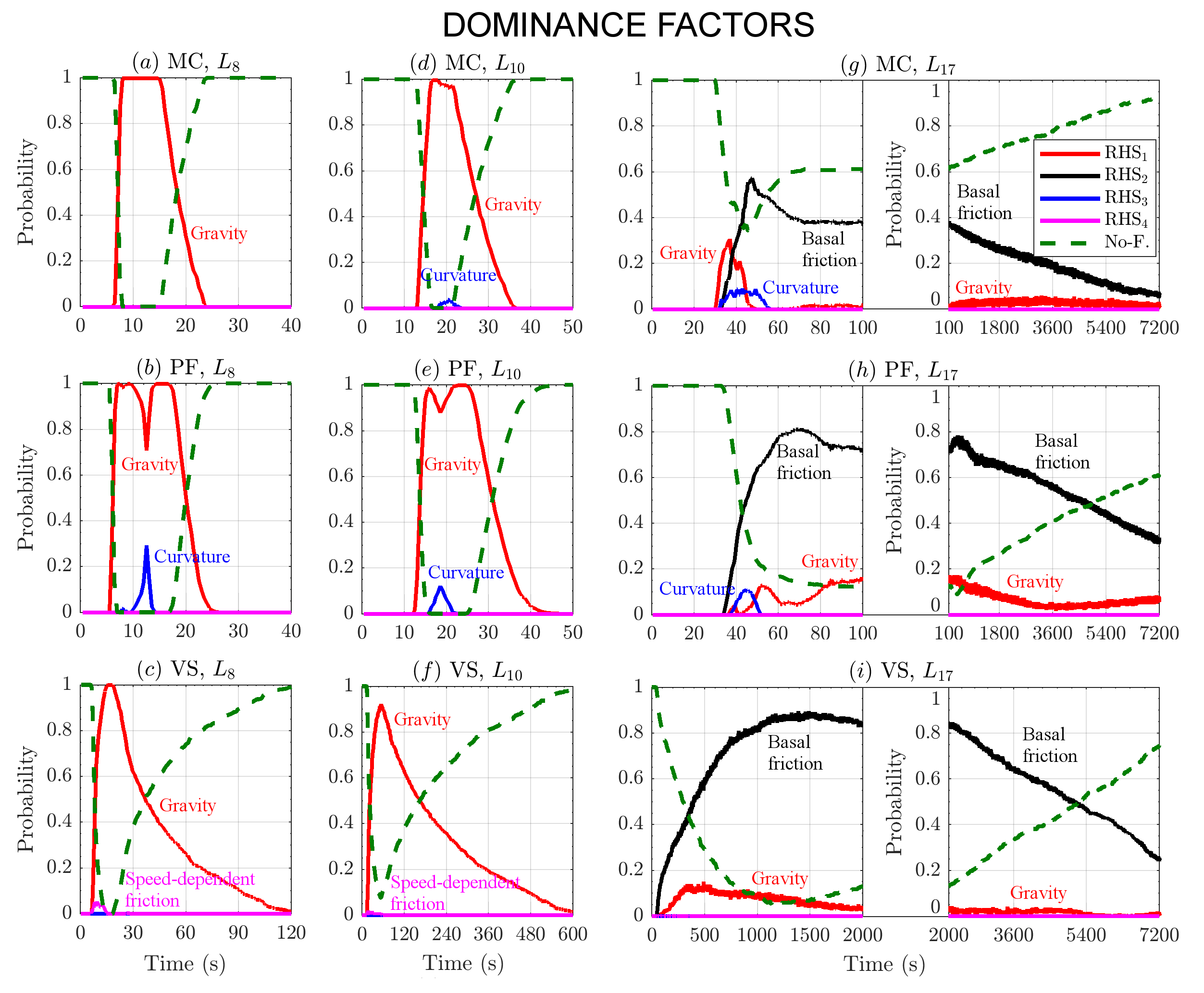}
        \caption{Dominance factors of the forces in three locations in the first km of runout. Different models are plotted separately: (a,d,g) assume MC; (b,e,h) assume PF; (c,f,i) assume VS. Different colors correspond to different force terms. No-flow probability is displayed with a green dashed line.}
        \label{fig:Colima-Pr1}
\end{figure}
The plots \ref{fig:Colima-Pr1}a,b,c and \ref{fig:Colima-Pr1}d,e,f are related to point $L_8$ and $L_{10}$, respectively. They are significantly similar. The gravitational force $\boldsymbol{RHS_1}$ is the dominant variable with a very high chance, $P_1>90\%$. In MC and PF there is a small probability, $P_3=5\%-30\%$, of $\boldsymbol{RHS_3}$ being the dominant variable for $\sim 5 s$. In VS it is observed a $P_4=5\%$ chance of $\boldsymbol{RHS_4}$ being dominant, just for a few seconds. Plots \ref{fig:Colima-Pr1}g,h,i concern the relatively more distal point $L_{17}$. They are split in two sub-frames at different time scale. In all the models, $\boldsymbol{RHS_2}$ is the most probable dominant variable, and its dominance factor has a bell-shaped profile. In all the models, also $\boldsymbol{RHS_1}$ has a small chance of being the dominant variable. In MC this chance is more significant, at most $P_1=30\%$ for $\sim 20 s$, and again $P_1=\%2$ in $[100, 7200] s$. In PF $P_1=15\%$ in two peaks, one short lasting at about $55 s$, and the second extending in $[100,500] s$. Also in VS, $P_1=15\%$ at $[300, 500] s$. Its profile is unimodal in time and becomes $P_1<2\%$ after $2000 s$. In MC and PF, $\boldsymbol{RHS_3}$ has a dominance factor $P_3=10\%$ at $[30, 50] s$ and $[40, 50] s$, respectively.

In summary, gravitational force is dominant with a very high chance until the no-flow probability becomes large. In MC and PF curvature related forces can also be dominant for a short time. In VS gravitational force is dominant for a larger time span than in the other models, because of the longer presence of the flow. The speed dependent friction can be dominant with a small probability at the beginning of the dynamics.
\begin{figure}[H]
         \centering
        \includegraphics[width=0.95\textwidth]{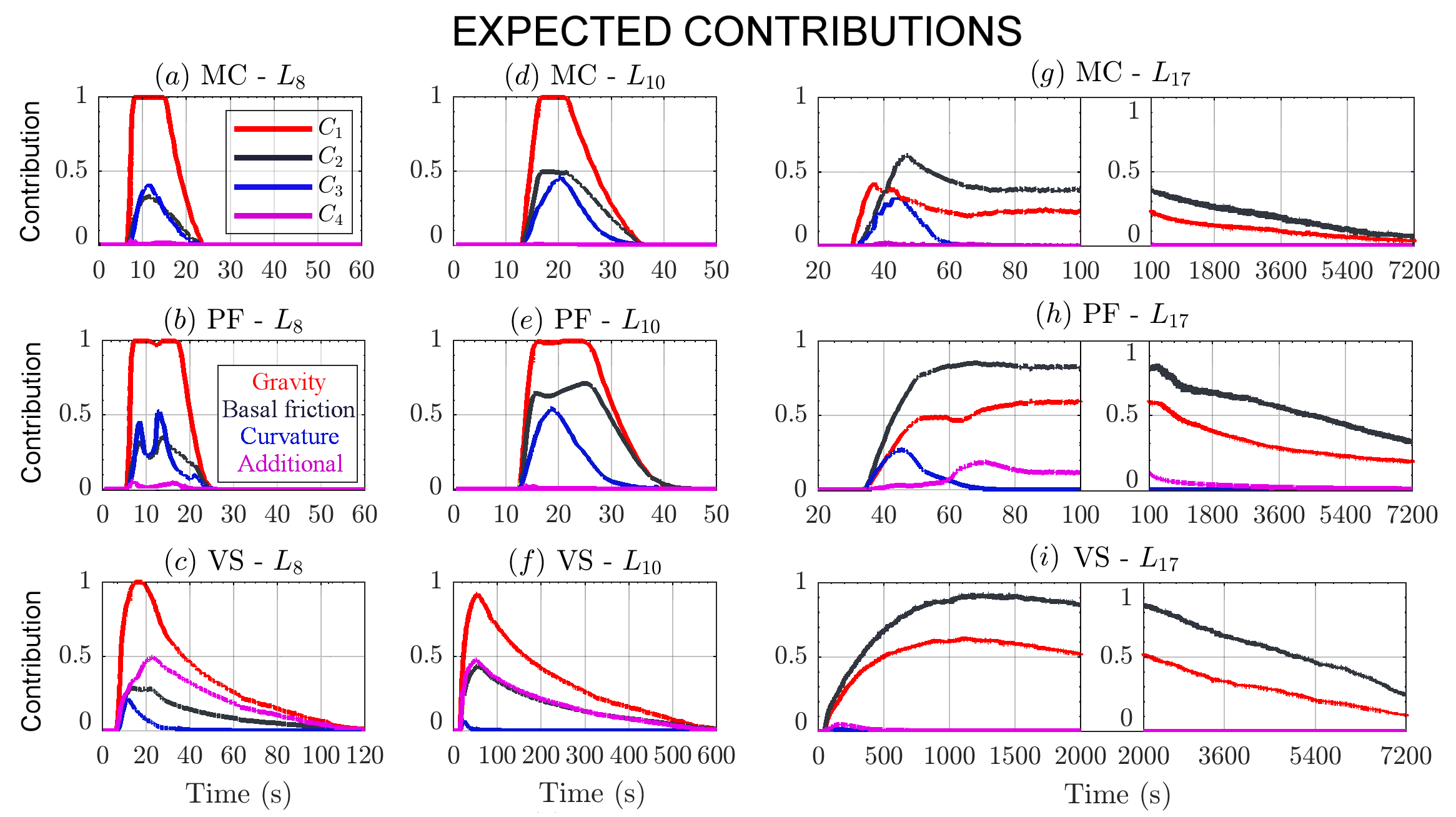}
        \caption{Expected contributions of the forces in three locations in the first km of runout. Different models are plotted separately: (a,d,g) assume MC; (b,e,h) assume PF; (c,f,i) assume VS. Different colors correspond to different force terms.}
        \label{fig:Colima-Ci_1}
\end{figure}
Figure \ref{fig:Colima-Ci_1} shows the corresponding expected contributions $\mathbb E[C_i]_{i=1,\dots,4}$. $\forall i$, $C_i$ is related to the force term $\boldsymbol{RHS_i}$. The contributions in points $L_8$ and $L_{10}$ are shown in \ref{fig:Colima-Ci_1}a,b,c and \ref{fig:Colima-Ci_1}d,e,f, respectively. The plots related to the same model are similar. In all the models $C_1$ is significantly larger than $C_2$ and $C_3$, which are almost equivalent in MC and PF, while $C_2>C_3$ in VS. $C_4$ always gives a negligible contribution, except in VS, where it is comparable to $C_2$. In $L_8$, following PF, $C_3$ is bimodal, whereas it is unimodal in MC and VS. This is not observed in $L_{10}$. In $L_8$, $C_3$ is greater than in $L_{10}$, compared to the other forces. VS always shows a slower decrease of the plots. In plots \ref{fig:Colima-Ci_1}g,h,i are shown the expected contributions in $L_{17}$. The plots are split in two sub-frames at different time scale. Initial dynamics is dominated by $C_2$, except for in MC, and only for a short time, $[30, 35] s$. In MC there is an initial peak of $C_2$ which is not observed in the other models. $C_3$ has a significant size, in MC and PF, and unimodal profile. In PF, after $C_3$ wanes, at about $60 s$ also $C_4$ becomes not negligible for $\sim 40 s$. The second part of the temporal domain is characterized by a slow decrease of $C_2>C_1$.

\subsubsection{Three locations distal from the initial pile}
Figure \ref{fig:Colima-Pr2} shows the dominance factors $(P_i)_{i=1,\dots,4}$ in the three distal points $L_{39}$, $L_{43}$, and $L_{46}$, all more than 2 km far from the initial pile. Plots \ref{fig:Colima-Pr2}a,b,c and $L_{39}$ are dominated by $\boldsymbol{RHS_1}$. In all the models $P_1$ is increasing and $P_1>90\%$ at the end of the simulation. In MC, $P_1$ shows a plateau at $\sim 40\%$ in $[90,2000] s$ preceded and followed by steep increases, while in the other models it rises gradually. $P_2>0$ after $\sim 500 s$ and $3600 s$, respectively, but is never greater than $2\%$. In MC $P_3 \approx 10\%$ at $[50,70] s$. No-flow probability becomes zero in PF and VS, while stops at $20\%$ in MC. Plots \ref{fig:Colima-Pr2}d,e,f are related to point $L_{43}$, and are remarkably complex. In MC, either $P_1$ and $P_2$ are $\sim 35\%$ in the first $200 s$. Then, $P_2$ increases, and $\boldsymbol{RHS_2}$ becomes the only dominant variable after $3600 s$. The no-flow probability is never below $30\%$. $P_3=35\%$ in $[40, 60] s$. Instead, in PF $P_1>90\%$ until $3600 s$, and $P_2$ rises only in the very last amount of time, reaching $P_2=P_1=40\%$. The no-flow probability is very low during the most of the temporal window, rising at $20\%$ only at $7200 s$. Both $P_3$ and $P_4$ show short peaks, $\sim 10\%$, at $[50,60] s$. In VS the no-flow probability is never below $20\%$, and the dominance factors are broadly equivalent to MC, although $P_1$ is the greatest up to $4000 s$, and $P_3\equiv0$. Plots \ref{fig:Colima-Pr2}g,h,i are related to point $L_{46}$ and they are similar to those recorded at point $L_{17}$, but $P_2>90\%$ and the no-flow probability decreases to zero in the second half of the simulation. Moreover, in all the models $P_1$ does not show any initial peak and instead increases slowly, reaching $P_1=10\%$ after more than $3600 s$.
\begin{figure}[H]
         \centering
        \includegraphics[width=0.95\textwidth]{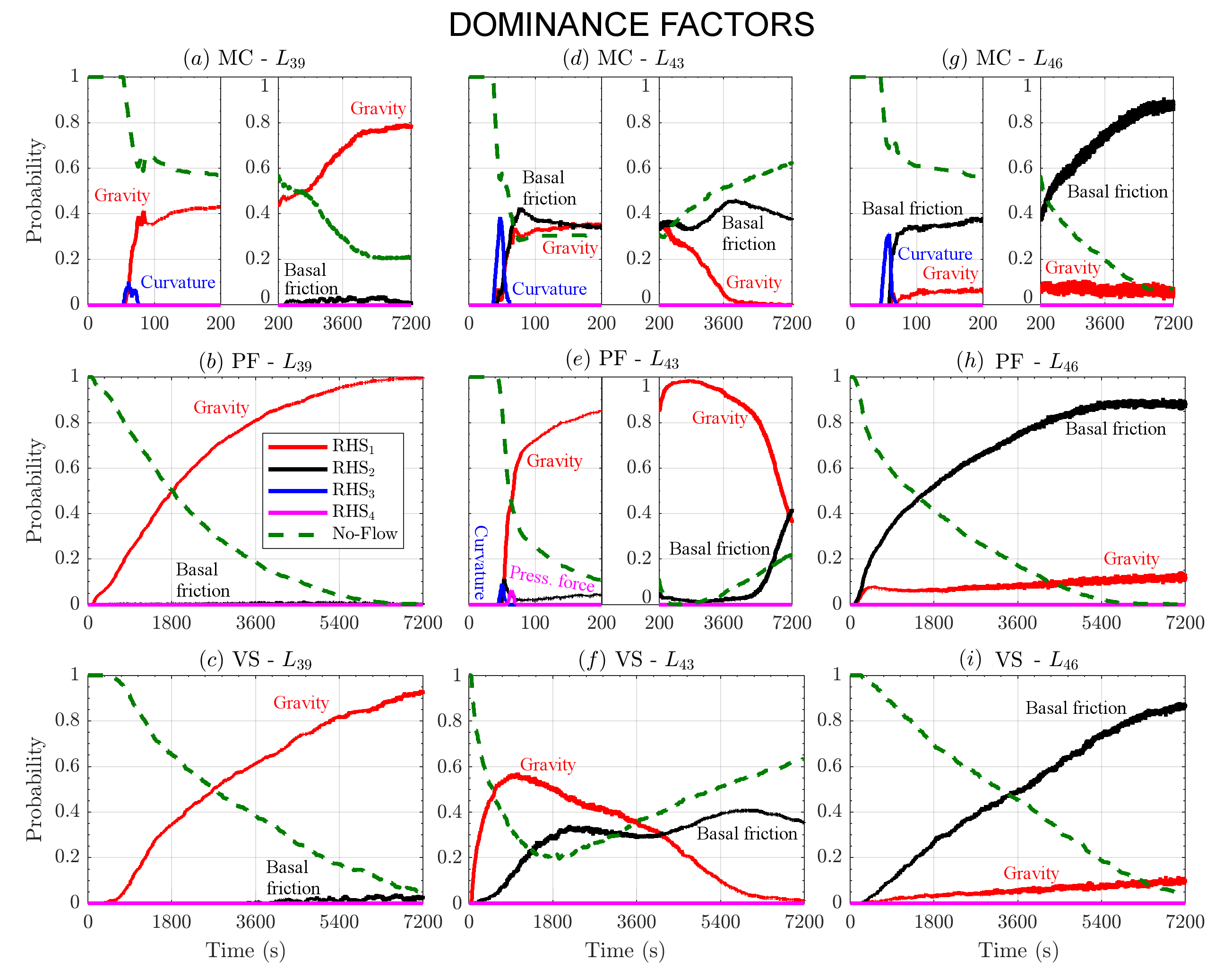}
        \caption{Dominance factors of the forces in three locations after 2 km of runout. Different models are plotted separately: (a,d,g) assume MC; (b,e,h) assume PF; (c,f,i) assume VS. Different colors correspond to different force terms. No-flow probability is displayed with a green dashed line.}
        \label{fig:Colima-Pr2}
\end{figure}
In summary, only the gravity or the basal friction are dominant with high probability. Some of the points have a deposit at the end with a high chance, some other not, depending on the slope. In general, in MC the no-flow probability tends to be larger than in the other models, because some flow samples stops earlier, or completely leaves the site. Again, curvature can have a small chance to be dominant in MC and PF, particularly when the speed is high. Point $L_{43}$ deserves a specific discussion. It is not proximal to the initiation, but the no-flow probability is increasing at the end, meaning that all the material tends to leave the site. Moreover, the dominating force can be the gravity or the basal friction depending on the time and the model. In MC and VS both the two forces have similar chances to be dominant for  most of the time of the simulation. In PF, only the gravitational force is dominant with a high chance. This is probably because point $L_{43}$ is situated downhill of a place where a significant amount of material stops.
\begin{figure}[H]
         \centering
        \includegraphics[width=0.95\textwidth]{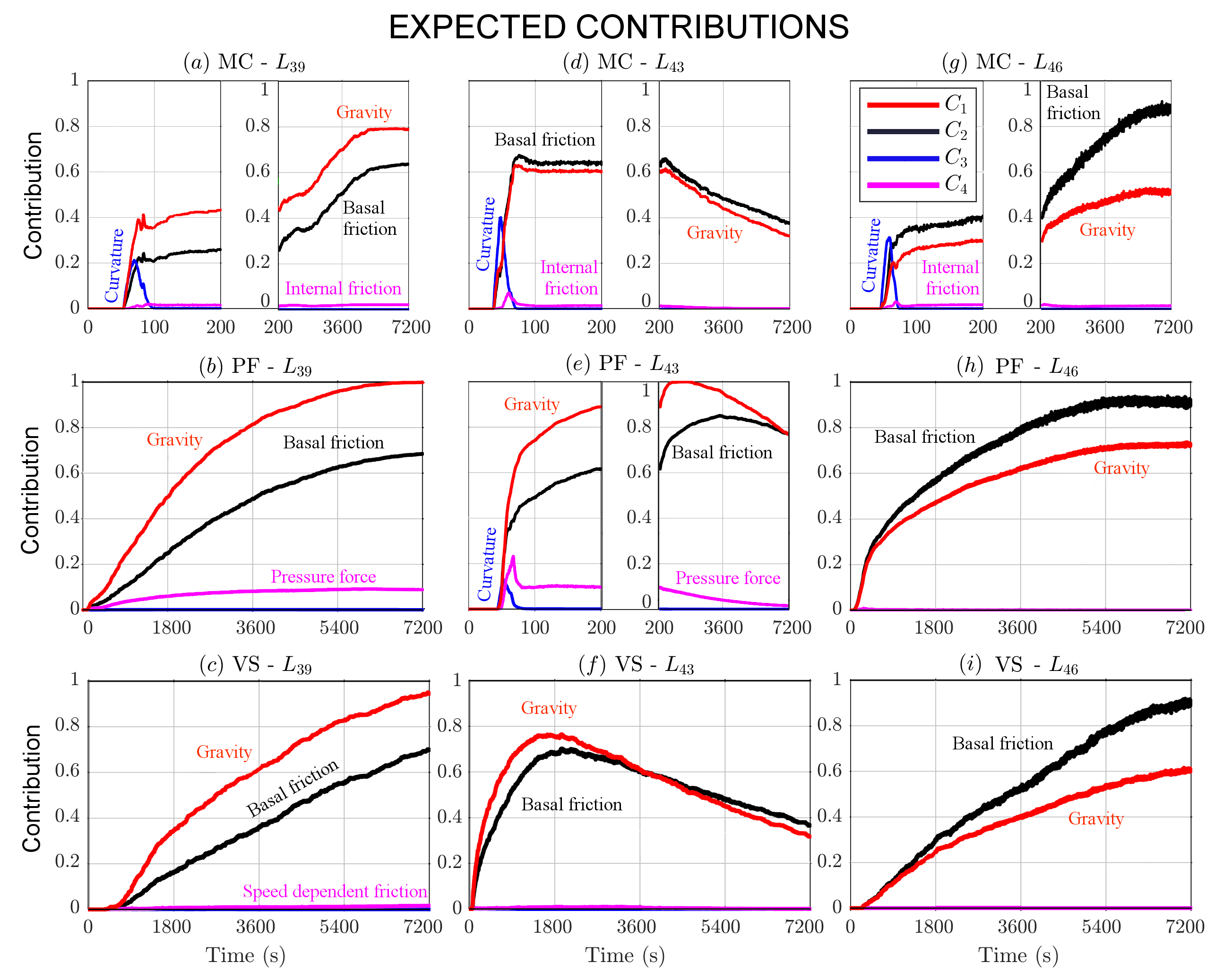}
        \caption{Expected contributions of the forces in three locations after 2 km of runout. Different models are plotted separately: (a,d,g) assume MC; (b,e,h) assume PF; (c,f,i) assume VS. Different colors correspond to different force terms.}
        \label{fig:Colima-Ci_2}
\end{figure}
Figure \ref{fig:Colima-Ci_2} shows the expected contributions in the distal points. In general, it is worth noting that the remarkable diversity in the dominance factors between the different locations can be the consequence of even a small imbalance between gravity and basal friction. All the plots are dominated by $C_1$ and $C_2$, and the remarkable differences observed in the dominance factors depend on which contribution is the greatest. In general these two contributions have similar profiles. Plots \ref{fig:Colima-Ci_2}a,b,c are related to point $L_{39}$ and $C_1>C_2$. In MC, also is $C_3>0$ for a short time. In MC and PF also $C_4>0$, but it is significantly lower than the previous contributions, almost negligible in MC. Plots \ref{fig:Colima-Ci_2}d,e,f concern point $L_{43}$. In MC $C_1<C_2$, in PF $C_1>C_2$, in VS they $C_1$ decreases and crosses $C_2$ at $\sim 3600 s$. The two contributions form a plateau in MC, in $[90, 200] s$. In MC and PF $C_3>0$ for a few seconds, and also $C_4>0$ with an initial spike at $\sim 60s$. In PF it shows a long lasting plateau, while becomes negligible in MC. Plots \ref{fig:Colima-Ci_2}g,h,i are related to point $L_{46}$. In all the models $C_1<C_2$, and these force contributions are monotone increasing. Only in MC $C_3>0$ shortly, and $C_4>0$, but almost negligible.

\subsection{Flow extent and spatial integrals}
Figure \ref{fig:Colima-spatial} shows the volumetric average of speed and Froude Number. It also shows the inundated area as a function of time. Like in Fig.\ref{fig:Ramp-spatial}, spatial averages and inundated area have smoother plots than local measurements, and most of the details observed in local measurements are not easy to discern. In plot \ref{fig:Colima-spatial}a the speed shows a bell-shaped profile in all the models, but whereas the all values were close in the inclined plane experiment, in this case the maximum speed is $\sim 60 m/s$ in MC, $\sim 50 m/s$ in PF, $\sim 20 m/s$ in VS, on average. Uncertainty is $\pm 18 m/s$ in MC, similar, but skewed, in VS, and $\pm 10 m/s$ in PF.
\begin{figure}[H]
        \centering
        \includegraphics[width=0.85\textwidth]{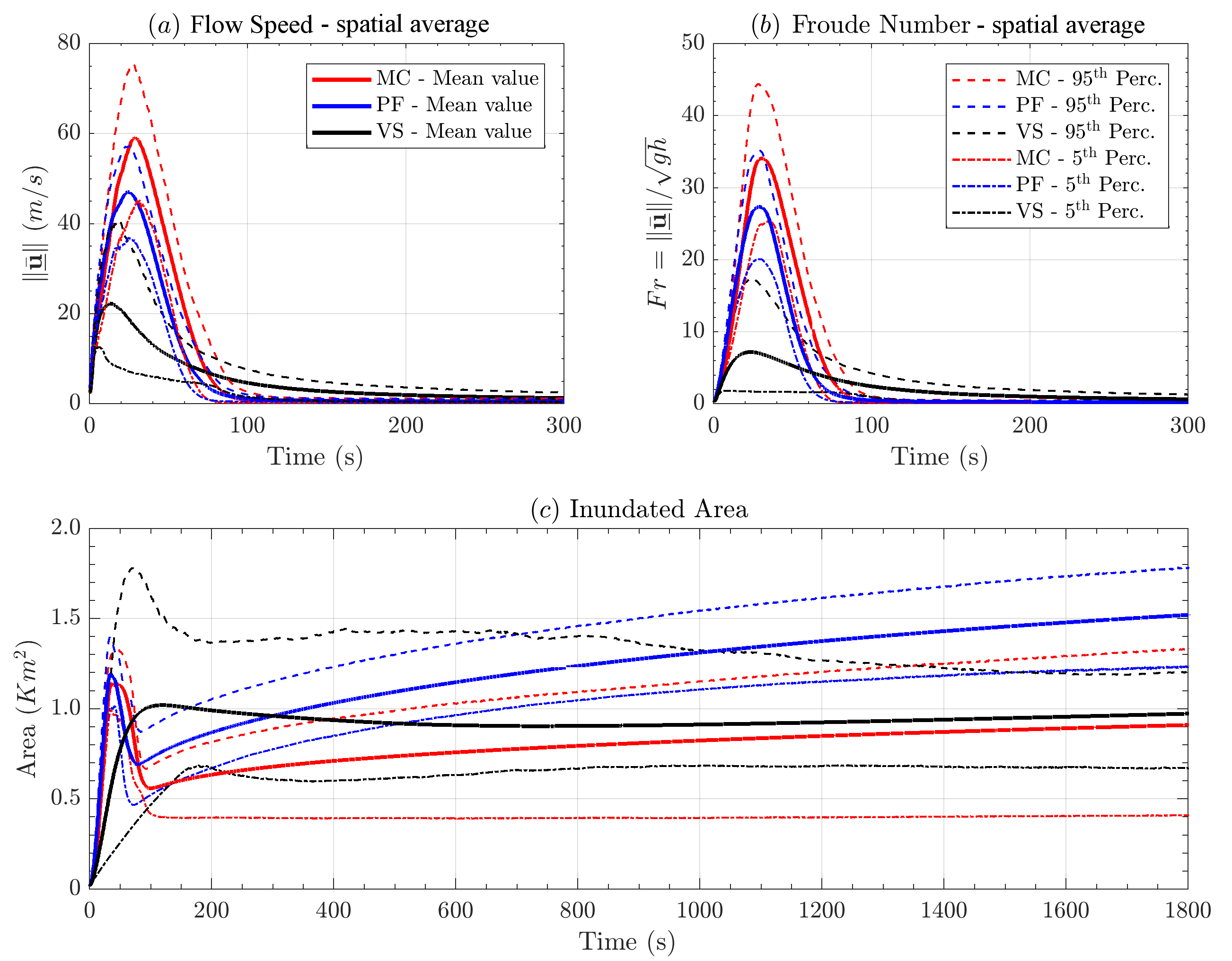}
        \caption{Comparison between spatial averages of $(a)$ flow speed, and $(b)$ Froude Number, in addition to the $(c)$ inundated area, as a function of time. Bold line is mean value, dashed/dotted lines are 5$^{\mathrm{th}}$ and 95$^{\mathrm{th}}$ percentile bounds. Different models are displayed with different colors.}
        \label{fig:Colima-spatial}
\end{figure}
In plot \ref{fig:Colima-spatial}b, the $Fr$ profile is very similar to the speed, but the difference between VS and the other models is accentuated. Maximum values are $\sim 50$ in MC, $\sim 38$ in PF, $\sim 5$ in VS, whereas uncertainty is $\pm 10$ in MC, $\pm 7$ in PF, and skewed $[-5, +10]$ in VS. In plot \ref{fig:Colima-spatial}c, inundated area has a first peak in MC and PF, both at $\sim 1.15 km^2$, followed by a decrease to $0.55 km^2$ and $0.7 km^2$, respectively, and then a slower increase up to a flat plateau at $0.9 km^2$ and $1.5 km^2$, respectively. Uncertainty is $\sim \pm 0.2 km^2$ in both MC and PF until $\sim 100 s$, and then it increases at $\pm 0.3 km^2$ and $[-0.5, +0.4] km^2$, respectively. In MC this increase in uncertainty is concentrated at $\sim 100 s$, while it is more gradual in PF. VS has a different profile. The initial peak is only significant in the 95$^{th}$ percentile values, and occurs later, at $\sim 100 s$. The peak is of $\sim 1 km^2$ on the average, but up to $\sim 1.8 km^2$ in the 95$^{th}$ percentile. The decrease after the peak is very slow and the average inundated area never goes below $0.85 km^2$, and eventually reaches back to $\sim 1 km^2$. Uncertainty is $[-0.3, +0.2] km^2$.

\subsection{Power integrals}
Figure \ref{fig:Colima-Power-spatial} shows the spatial sum of the powers. The estimates in this section assume $\rho = 1800 kg/m^3$ as a constant scaling factor. Corresponding plots of the force terms are included in Supporting Information S6. The scalar product of force with velocity imposes the bell-shaped profile already observed in Fig. \ref{fig:Ramp-Power-spatial}a. In general, gravity term is larger in VS, because a portion of the flow lingers on the higher slopes for a long time. Basal friction has a higher peak in PF compared to the other models, due to the interpolation of the two basal friction angles.
\begin{figure}[H]
        \centering
        \includegraphics[width=0.90\textwidth]{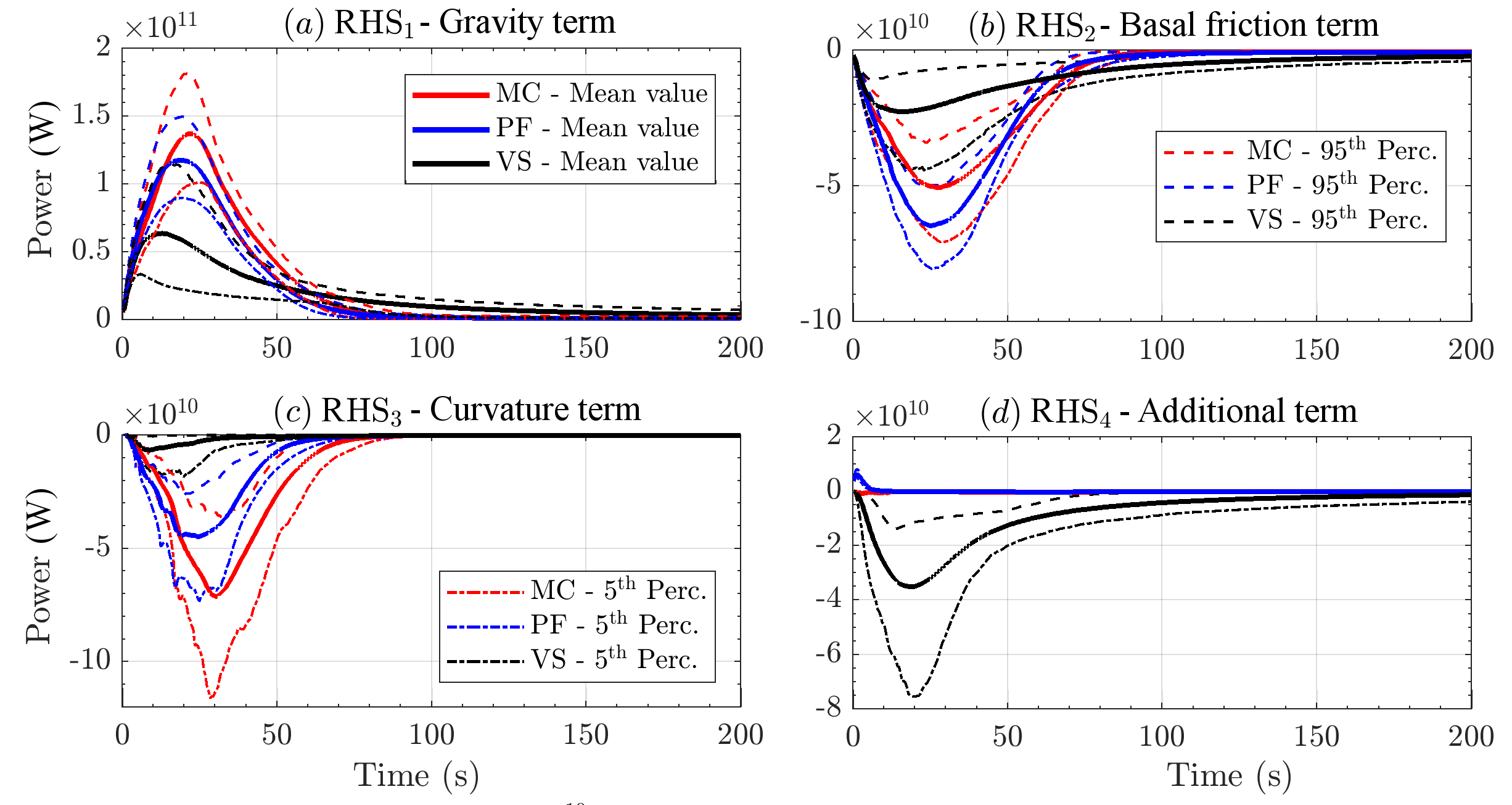}
        \caption{Spatial integrals of the powers. Bold line is mean value, dashed lines are 5$^{\mathrm{th}}$ and 95$^{\mathrm{th}}$ percentile bounds. Model comparison on the mean value is also displayed. Different models are displayed with different colors.}
        \label{fig:Colima-Power-spatial}
\end{figure}
In plot \ref{fig:Colima-Power-spatial}a the power of $\boldsymbol{RHS_1}$ starts from zero and rises up to $\sim 1.4e11 W$ in MC, $\sim 1.2e11 W$ in PF, $\sim 6.5e10 W$ in VS. Uncertainty is $\pm 4e10 W$ in MC, $\pm 3e10 W$ in PF, $[-4e10,+5e10] W$ in VS. The decrease of gravitational power is related to the slope reduction, and this decrease is more gradual in VS than in the other models. In plot \ref{fig:Colima-Power-spatial}b the power of  $\boldsymbol{RHS_2}$ is always negative and peaks to $\sim -6.5e10 W$ in MC, $\sim -5e10 W$ in PF, $\sim -2e10 W$ in VS. In VS this dissipative power is significantly more flat than in the other models. MC and PF show negligible powers after $\sim 100 s$, VS after $\sim 200 s$. Uncertainty is $\pm 2e10 W$ in MC, $\pm 1.5e10 W$ in PF, $[-2e10,+1e10] W$ in VS. In PF, the plot starts from stronger values than in the other models, but it is also the faster to wane. In plot \ref{fig:Colima-Power-spatial}c the power of $\boldsymbol{RHS_3}$ shows a negative peak at $\sim -7e10 W$ in MC, $\sim -4.5e10 W$ in PF, $\sim -5e9 W$ in VS. Uncertainty on the peak value is $[-4.5e10,+3.5e10] W$ in MC, $[-2.5e10,+2e10] W$ in PF, $[-1e10,+5e9] W$ in VS. The three models all show a bell-shaped profile, MC and PF waning to zero at $90 s$, VS at $\sim 30 s$. In plot \ref{fig:Colima-Power-spatial}d the power of $\boldsymbol{RHS_4}$ has a different meaning in the three models. In MC it is the internal friction term, and it only has almost negligible ripple visible in the first second. In PF it is a depth averaged pressure force linked to the thickness gradient, and has a very small effect limited to the first second of simulation, at $5e9 W$. It becomes null at $\sim 10 s$. In VS, instead, it is a speed dependent term, and has a very relevant effect. The plot shows a bell-shaped profile, with a peak of $\sim -3.5e10 W$, $[-2e10,+1e10] W$. After that, this dissipative power gradually decreases, and becomes negligible at $200 s$.

\subsection{Example of model performance calculation}
Finally, we give an example of model performance evaluation of the couple $\left(M, P_M\right)$ according to a specific observation. In past work \citep{Patra2005}, MC rheology was tuned to match deposits for known block and ash flows, but {\it a priori} predictive ability was limited by inability to tune without knowledge of flow character. The new procedure developed in this study enables an enhanced quantification of model performance, i.e. the similarity of the outputs and real data. We remark that the measured performance refers to the couple $\left(M, P_M\right)$, and that different parameter ranges can produce different performances \citep{Tierz2016}. This is in contrast with traditional performance analysis based on particular, albeit calibrated, simulations \citep{Charbonnier2012}.

Our example concerns the Volc{\'a}n de Colima case study, and in particular we compare the inundated region in our simulations to the deposit of a BAF occurred 16 April 1991 \citep{Saucedo2004, Rupp2004, Rupp2006}. The inundated region is defined as the points in which the maximum flow height $H$ is greater than $10 cm$. A similar procedure may be applied to any observed variable produced by the models, if specific data become available. Let $\mathcal M:\mathcal P(\mathbb R^2)\rightarrow[0,1]$ be a similarity index defined on the subsets of the real plane. An equivalent definition can be based on the pseudo-metric $1-\mathcal M$. For example, we define
$$\mathcal M_I:=\frac{\int_{\mathbb R^2} 1_{S \cap D}(\textbf{x}) d\textbf{x}}{\int_{\mathbb R^2} 1_D(\textbf{x})d\textbf{x}},\quad \mathcal M_U:=\frac{\int_{\mathbb R^2} 1_D(\textbf{x})d\textbf{x}}{\int_{\mathbb R^2} 1_{S \cup D}(\textbf{x}) d\textbf{x}}, \quad \mathcal J:=\mathcal M_I\cdot \mathcal M_U,$$
where $S\subset \mathbb R^2$ is the inundated region, and $D\subset \mathbb R^2$ is the recorded deposit. In particular, $\mathcal M_I$ is the area of the intersection of inundated region and deposit over the area of the deposit, $\mathcal M_U$ is the area of the deposit over the area of the union of inundated region and deposit, $\mathcal J$ is the product of the previous, also called Jaccard Index \citep{Jaccard1901}.

Figure \ref{fig:Colima-Hist} shows the probability distribution of the similarity indices, according to the uniform probability $P_M$ on the parameter ranges defined in this study. Different metrics can produce different performance estimates, for example MC inundates most of the deposit, but overestimates the inundated region, while VS relatively reduces the inundated region outside of the deposit boundary, but also leaves several not-inundated spots inside it.

\begin{figure}[H]
         \centering
        \includegraphics[width=0.95\textwidth]{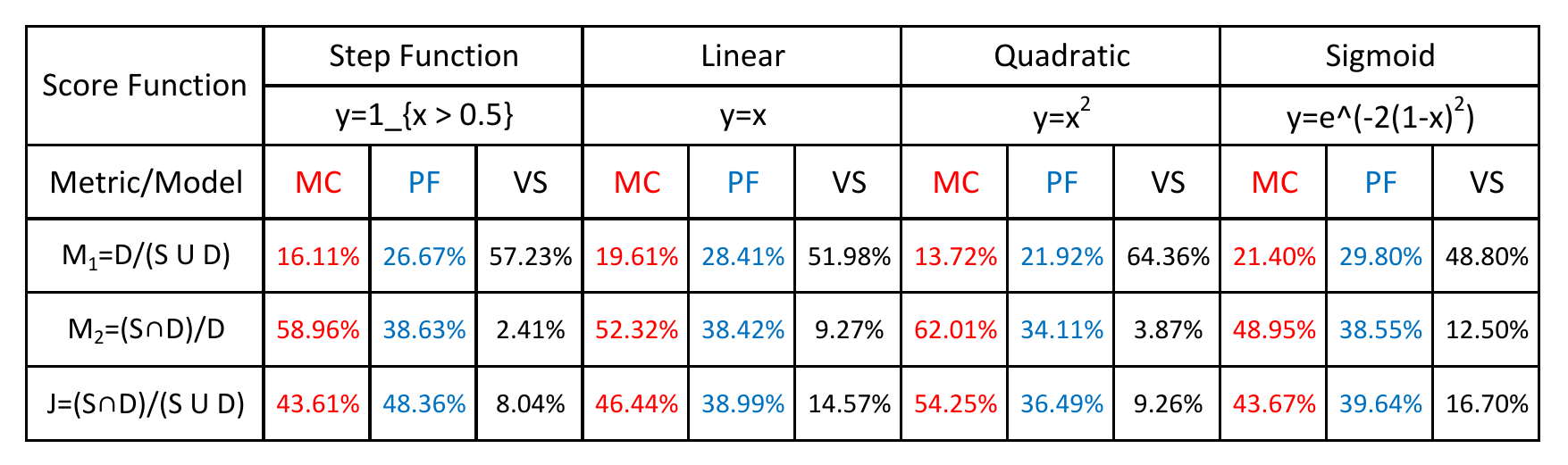}
        \caption{Performance scores as a function of model, performance metric and score function.}
\end{figure}

Let $g:[a,b]\rightarrow[0,1]$ be a score function defined over the percentile range of the similarity index. The global $5^{th}$ and $95^{th}$ percentile values $[a,b]$ are defined assuming to select the model randomly with equal chance, and are also shown in Fig.\ref{fig:Colima-Hist}a,b,c.

Then the performance score $G_g$ of model $\left(M, P_M\right)$ is defined as:
$$G_g\left(M, P_M\right)=\int_{[a,b]} g(x) df_M(x),$$
where $f_M$ is the pdf related to the model. Possible score functions include a step function at the global median, a linear or quadratic function, a sigmoid function. Table 1 shows alternative performance scores, according to changing similarity indices and score functions.

\begin{figure}[H]
         \centering
        \includegraphics[width=0.9\textwidth]{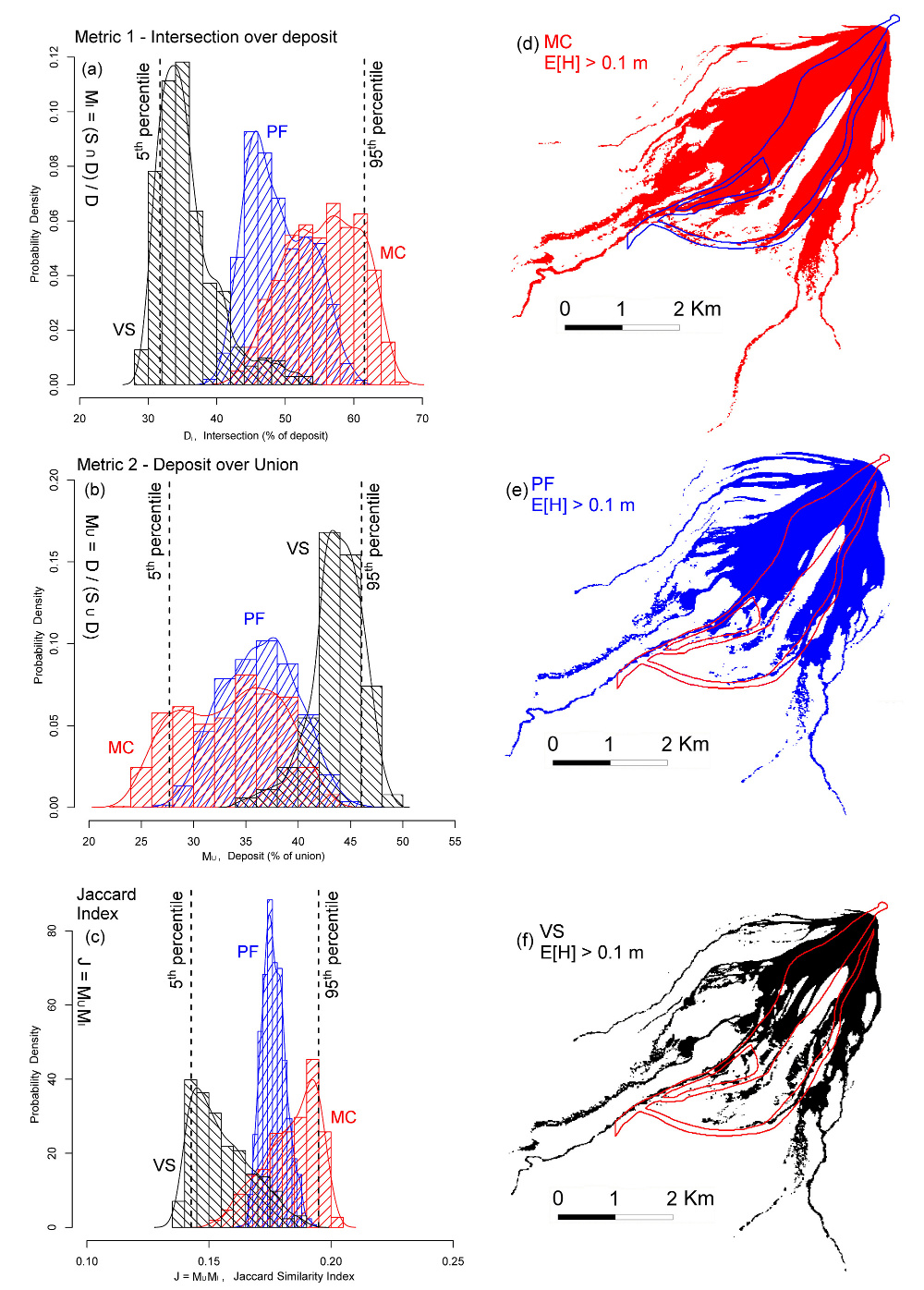}
        \caption{Volc{\'a}n de Colima. Pdf of the similarity index of inundated regions and a real BAF deposit. (a) is based on $\mathcal M_I$, (b) on $\mathcal M_U$, (c) on $\mathcal J$. The models are MC (red), PF (blue), VS (black). Data histograms are displayed in the background. Global $5^{th}$ and $95^{th}$ percentile values are indicated with dashed lines. Plots (d,e,f) display the average inundated region $\left\{\textbf{x} : E[H(\textbf{x}]>10 cm\right\}$. The boundary of the real deposit is marked with a colored line.}\label{fig:Colima-Hist}
\end{figure}

\section{Conclusions}
In this study, we have introduced a new statistically driven method for analyzing complex models and their constituents. We have used three different models arising from different assumptions about rheology in geophysical mass flows to illustrate the approach. The data shows unambiguously the performance of the models across a wide range of possible flow regimes and topographies. The analysis of contributing variables is particularly illustrative of the impact of modeling assumptions. Knowledge of which assumptions dominate, and by how much, allows us to construct efficient models for desired inputs. Such model composition is the subject of ongoing and future work, with the purpose of bypassing the search for a unique best model, and going beyond a simple mixture of alternative models.

In summary, our new method enabled us to break down the effects of the different physical assumptions in the dynamics, providing an improved understanding of what characterizes each model. The procedure was applied to two different case studies: a small scale inclined plane with a flat runway, and a large scale DEM on the SW slope of Volc\'{a}n de Colima (MX). In particular, we presented:
\begin{itemize}
  \item A short review of the assumptions characterizing three commonly used rheologies of Mohr-Coulomb, Poliquenne-Forterre, Voellmy-Salm. This included a qualitative list of such assumptions, and a breaking down of the different terms in the differential equations.
  \item A new statistical framework, processing the mean and the uncertainty range of either observable or contributing variables in the simulations. The new concepts of dominance factors and expected contributions enabled a simplified description of the local dynamics. These quantities were analyzed at selected sites, and spatial integrals were calculated, which illustrate the characteristics of the entire flow.
  \item The contribution  coefficients $C_i$ and dominance factors $P_i$  introduced here allow us to quantify and compare in a probabilistic framework the effect of modeling assumptions based on the full range of flows explored using statistically rigorous ensemble computations.
  \item A final discussion, explaining all the observed features in the results in light of the known physical assumptions of the models, and the evolving flow regime in space and time. This included an example of the model performance estimation method, which depends on the metric and the cost function adopted.
\end{itemize}

Our analysis uncovered the following main features of the different geophysical models used in the example analysis:
\begin{itemize}
  \item Compared to the standard MC model, the lack of internal friction in the PF model produces an accentuated lateral spread. The spread is increased by the uninhibited internal pressure force, which briefly pushes the flow ahead and laterally during the initial collapse. That force can also have some minor effects in the accumulation of the final deposit. The interpolation between the smaller bed friction angle $\phi_1$ and the larger value $\phi_2$ in the PF model suddenly stops the flow if it is thin compared to its speed. This mechanism suppresses large peaks in flow speed.
  \item In VS, the speed-dependent friction has a great effect in reducing lateral spread and producing channelized flow even where there are otherwise minor ridges and adverse slopes in the topography. The flow tends to be significantly slower and more stretched out in the downslope direction. The effects of different formulations of the curvature term have less impact than do the effects of lower basal friction and speed.
\end{itemize}

Furthermore, we can make the following statements about the technique in terms of its use on models in general:
\begin{itemize}
\item It gives information not only on which forces in the equations of motion are dominating the flow, but also shows where these forces are greatest and gives insight into why they locally peak and vary, and into the impacts of the dominating forces on the model flow outputs,
\item It provides a new, quantitative technique to evaluate the most important forces or phenomena acting in a particular model domain, which can supplement, provide insight and guidance into, and generate quantitative information for, the more typical methods used in force analysis of intuition and Similarity Theory.
\end{itemize}

Additional research concerning other case studies, and different parameter ranges, might reveal other flow regimes, and hence differences in the consequences of the modeling assumptions under new circumstances.

\paragraph{Acknowledgements}
We would like to acknowledge the support of NSF awards 1521855, 1621853, and 1339765.

\appendix
\section{Latin Hypercubes based on Orthogonal Arrays}\label{A-1}
The Latin Hypercube Sampling is a well established procedure for defining pseudo-random designs of samples in $\mathbb R^d$, with good properties with respect to the uniform probability distribution on an hypercube $[0,1]^d$ \citep{McKay1979,Owen1992b,Stein1987,Ranjan2014,Mingyao2016}. In particular, compared to a random sampling, a Latin Hypercube: (i) enhances the capability to fill the d-dimensional space with a finite number of points, (ii) in case $d>1$, avoids the overlapping of point locations in the one dimensional projections, (iii) reduces the dependence of the number of points necessary on the dimensionality $d$.


The procedure is simple: once the desired number of samples $N\in\mathbb N$ is selected, and $[0,1]$ is divided in $N$ equal bins, then each bin will contain one and only one projection of the samples over every coordinate. The Latin Hypercube Sampling definition is trivially generalized over $C=\prod^d_i [a_i, b_i]$, i.e. the cartesian product of $d$ arbitrary intervals. This has been applied in this study, to define a Latin Hypercube Sampling over the parameter domain of the flow models.

There are a large number of possible designs, corresponding the number of permutations of the bins in the d-projections, i.e. $d\cdot N!$. If the permutation is randomly selected there is a high possibility that the design will have good space filling properties. However, this is not assured, and clusters of points or regions of void space may be observed in $C$. For this reason, we based our design on the Orthogonal Arrays (OA) theory \citep{Owen1992a,Tang1993}.

\begin{definition}[Orthogonal arrays]
Let $S=\{1,\dots,s\}$, where $s\ge 2$. Let $Q\in S^{n\times m}$ be a matrix of such integer values. Then $Q$ is called an $OA(n,m,s,r)$ $\Longleftrightarrow$ each $n\times r$ submatrix of $Q$ contains all possible $1\times r$ row vectors with the same frequency $\lambda=n/s^r$, which is called the index of the array. In particular, $r$ is called the strength, $n$ the size, $(m\ge r)$ the constrains, and $s$ the levels of the array.
\end{definition}

OA are very useful for defining special Latin Hypercubes which are also forced to fill the space (or its r-dimensional subspaces, for a chosen $r<d$) in a more robust way, at the cost of potentially requiring a larger number of points than in a traditional Latin Hypercube sampling. In particular, inside each r-dimensional projection, the OA-based design fills the space like a regular grid at a coarse scale, but it is still an Latin Hypercube Sampling at a fine scale. A complete proof can be found in \citep{Tang1993} and it is a straightforward verification of the required properties.

Dealing with relatively small $d$, i.e. $d\in\{3,4\}$, we adopt a Latin Hypercube Sampling based on a $OA(s^d,d,s,d)$. We take $s=8$ for the 3-dimensional designs over the parameter space of Mohr-Coulomb and Voellmy-Salm models, i.e. $512$ points; we took $s=6$ for the 4-dimensional designs over the more complex parameter space of the Pouliquen-Forterre model, i.e. 1296 points.

\section{Conditional decomposition of expected contributions}\label{A-2}
Expected contributions are obtained after diving the force terms by the dominant variable $\Phi$, which is an unknown quantity depending on time, location, and input parameters. Thus we provide an additional result, further explaining the meaning of those contributions through the conditional expectation.

\begin{proposition}
Let $(F_i)_{i\in I}$ be random variables on $(\Omega, \mathcal F, P)$. For each $i$, let $C_i$ be the random contribution of $F_i$. Then we have the following expression:
$$E[C_i]=\sum_j p_j \mathbb E\left[\frac{F_i}{|F_j|}\ \Big{|}\ \Phi=|F_j|\right],$$
where $p_j:=P\left\{\Phi=|F_j|\right\}$.
\end{proposition}

\begin{proof}
Let $Z$ be a discrete random variable such that, for each $j\in\mathbb N$, $(Z=j) \Longleftrightarrow (\Phi=|F_j|)$. Then, by the rule of chain expectation:
$$E[C_i]=\mathbb E\left[\frac{F_i}{\Phi}\right]=\mathbb E\left[\mathbb E\left[\frac{F_i}{\Phi}\ \Big{|}\ Z=j\right]\right]=$$
$$=\mathbb E\left[\mathbb E\left[\frac{F_i}{|F_j|}\ \Big{|}\ Z=j\right]\right]=\sum_j P\left\{Z=j\right\} \mathbb E\left[\frac{F_i}{|F_j|}\ \Big{|}\ Z=j\right].$$
Moreover, by definition, $p_j=P\left\{Z=j\right\}$. This completes the proof.
\end{proof}

Consequently, we define the conditional contributions.

\begin{definition}[Conditional contributions]
Let $(F_i)_{i=1,\dots, k}$ be random variables on $(\Omega, \mathcal F, P)$. Then, $\forall (i,j)$, the conditional contribution $C_{i,j}$ is defined as:
$$C_{i,j}:=\mathbb E\left[\frac{F_i}{|F_j|}\ \Big{|}\ \Phi=|F_j|\right],$$
where $\Phi$ is the dominant variable.
\end{definition}

\bibliographystyle{apalike}
\bibliography{mybibfile}

\end{document}